\documentclass{article}
\usepackage{amssymb}
\usepackage{graphicx}
\usepackage[singlespacing]{setspace}
\usepackage{amsmath}
\usepackage[header]{appendix}

\newtheorem{theorem}{Theorem}
\newtheorem{acknowledgement}{Acknowledgement}

\newtheorem{claim}{Claim}
\newtheorem{conclusion}{Conclusion}
\newtheorem{condition}{Condition}

\newtheorem{definition}{Definition}

\newtheorem{lemma}{Lemma}

\newenvironment{proof}[1][Proof]{\textbf{#1.} }{\ \rule{0.5em}{0.5em}}
\input{tcilatex}

\begin{document}

\title{The Periodic Joint Replenishment Problem is Strongly $\mathcal{NP}$%
-Hard}
\date{}
\author{Cohen Tamar$^{1}$, Yedidsion Liron$^{2}$\thanks{%
Corresponding author. Email - lirony@ie.technion.ac.il} \\
$^{1}$Operation Research Center,\\
MIT- Massachusetts Institute of Technology,\\
MA, USA.\\
$^{2}$Faculty of Industrial Engineering and Management,\\
Technion - Israel Institute of Technology,\\
Haifa, Israel.}
\maketitle

\begin{abstract}
In this paper we study the long-standing open question regarding the
computational complexity of one of the core problems in supply chains
management, the periodic joint replenishment problem. This problem has
received a lot of attention over the years and many heuristic and
approximation algorithms were suggested. However, in spite of the vast
effort, the complexity of the problem remained unresolved. In this paper, we
provide a proof that the problem is indeed strongly $\mathcal{NP}$-hard.
\end{abstract}

\section{Introduction}

Many inventory models are aimed at minimizing ordering and holding costs
while satisfying demand. The Joint Replenishment Problem (JRP) deals with
the prospect of saving resources through coordinated replenishments in order
to achieve substantial cost savings. In this research we study the
complexity of JRP. In the JRP one is required to schedule the replenishment
times of numerous commodities (sometimes called items or products) in order
to supply an external demand per commodity. We refer to the schedule of the
replenishment times as the ordering policy. Each commodity incurs fixed
ordering costs every time it is replenished as well as linear holding costs
that are proportional to the quantity of the commodity held in storage.
Linking all commodities, a joint ordering cost is incurred whenever one or
more commodities are ordered. The objective of JRP is to minimize the sum of
ordering and holding costs. It is a natural extension of the classical
economic lot-sizing model that considers the optimal trade-off between
ordering costs and holding costs for a single commodity. With multiple
commodities, JRP adds the possibility of saving resources via coordinated
replenishments, a common phenomenon in supply chain management. JRP is a
special case of One-Warehouse-N-Retailers problem (OWNR), which deals with a
single warehouse receiving goods from an external supplier and distributing
to multiple retailers. The warehouse could also serve as a storage point.
JRP in particular is a special case of the OWNR with a very high warehouse
holding cost.

There are some distinctions between variations of JRP.

\begin{itemize}
\item Commodity order policy constraints: There are 3 types of order policy
constraints for the JRP. The first model requires a periodic ordering
policy. A periodic ordering policy is one in which for each commodity we
must determine a cycle time. An order will occur at each\ multiple of that
cycle time. We refer to this model as the periodic JRP (PJRP). The second
model does not require a cycle time for each commodity; however it requires
a cyclic ordering policy. We refer to this model as the cyclic JRP (CJRP).
The last model has no limits on the ordering policy. Note that PJRP is a
constrained version of CJRP, which in turn is a constrained version of the
ordering policy JRP. In this research we focus on PJRP.

\item Joint order policy constraints: The joint ordering cost in the PJRP
model is a complicated function of the inter-replenishment times, so it is
often assumed that joint orders are placed periodically, even if some joint
orders are empty, and that the cycle times of the commodities are always a\
multiple of the joint order cycle time. We denote this type of policy as\
the General Integer model (GI), while policies with no joint order
constraints are referred to as\ the General Integer model with Correction
Factor (GICF). The PJRP with GI policy constraint is also referred to as
Strict PJRP. We refer to PJRP with GICF policy constraint as the General
PJRP or simply as PJRP. Note that Strict PJRP is a constrained version of
General PJRP. In this research we focus on the General PJRP.

\item Demand type: Another important distinction is between problems with
stationary demand for each commodity and problems with fluctuating demand.
Note that\ the problem with stationary demand is a special case of the
problem with fluctuating demand. In this research we focus on the set of
problems with stationary demand.

\item Time horizon: The time horizon defines the horizon for which one must
plan an order policy. We distinguish between the problem with infinite
horizon and the problem with finite horizon. Most of the research over the
years focused on JRP with an infinite horizon, specifically when considering
stationary demand. The motivation for considering a finite horizon with
non-stationary demand comes from knowing the demand only for a finite
horizon. This motivation does not apply for stationary demand.

\item Solution integrality: The integrality of the solution determines
whether the ordering policy will be integral or not. Note that the integral
problem is a constrained version of the continuance problem. In this
research we focus on the integral problem.
\end{itemize}

\subsection{Literature review}

As far as we know, the complexity of JRP with stationary demands remained
open for all models. Some papers addressing JRP with constant demands (e.g., 
\cite{LY2003}, \cite{MC2006}) mistakenly cite a result by Arkin et al. \cite%
{ARJ1989}, which proves that JRP with non-stationary demands is $\mathcal{NP}
$-hard. Arkin et al. \cite{ARJ1989} stated that since JRP is a special case
of OWNR, proving JRP hardness also proves the hardness of OWNR. Lately,
Schulz and Telha \cite{ST2011} have proved that finding an optimal
replenishment policy for the stationary PJRP is at least as hard as the
integer factorization problem. When referring to the existence of a
polynomial-time optimal algorithm for the stationary PJRP, Schulz and Telha
also stated that this case remains open. In this paper we show that the PJRP
with stationary demands is strongly $NP$-hard.

\textbf{Strict PJRP.} The problem of Strict PJRP was well covered in the
reviews by Goyal and Satir \cite{GS1989} and Khouja and Goyal \cite{KG2008}.
Many research attempts have been made to find efficient solutions to Strict
PJRP. In the early 1970s, two pioneer studies suggested\ a graphical
heuristic approach \cite{S1971}, \cite{N1973}. At the same time, Goyal \cite%
{G1973} had suggested a non-polynomial lower bound based heuristics to find
the optimal strict cyclic policy, in which the cycle time for each commodity
is the joint replenishment cycle time. Van Eijs \cite{V1993} suggested a
modified version of Goyal's algorithm that involved using a non-strict
cyclic policy.

Based on these studies, many heuristics have been developed to solve the
Strict PJRP. Silver \cite{S1976} developed a heuristic algorithm to find the
joint period cycle time. Following this algorithm many iterative search
heuristics were suggested with different search bounds (Kaspi and Rosenblatt 
\cite{KR1983}, Goyal and Belton \cite{GB1979}, Kaspi and Rosenblatt \cite%
{KR1991}, Goyal and Deshmukh \cite{GD1993}, Viswanathan \cite{V1996}, Fung
and Ma \cite{FM2001} that was later modified by \cite{v2002}, \cite{PD2004}%
). Wildeman et al. \cite{WFD1997} used the idea of the iterative search and
implemented it in a heuristic that converges to an optimal solution. For
certain values of the joint period cycle time they solved Strict PJRP
optimally using a Lipschitz optimization procedure. Another heuristic
approach for the problem was developed by Olsen \cite{O2005}, called
evolutionary algorithm.

Since JRP is a special case of OWNR, results regarding the OWNR hold for JRP
as well. Hence, the following results are applicable for JRP. A prominent
advancement in the study of OWNR, the optimal Power-of-Two policy, was
achieved by Roundy \cite{R1985}. This policy could be computed in $O(n\log
n) $ time. Roundy proved that the cost of the best power-of-two policy can
achieve $98\%$ of an optimal policy ($94\%$ if the base planning period is
fixed). In other words, he suggested a 1.02-approximation (1.064 for the
fixed based planning period) for JRP, where a $\rho $-approximation
algorithm is an algorithm that is polynomial with respect to the number of
elements, and the ratio between the worst case scenario solution and the
optimal solution is bounded by a constant, $\rho $. Note that fixed based
planning period implies an integral model, while the general power-of-two
policy allows a non-integral solution. Based on Roundy's findings, Jackson
et al. \cite{J1985} proposed an efficient algorithm that offers a
replenishment policy in which the cost is within\ a factor of $\sqrt{\frac{9%
}{8}}\approx 1.06$ of the optimal solution. This approximation was later
improved to $\frac{1}{\sqrt{2}\log 2}$ for a non-fixed based planning period 
\cite{mr1993}.

Several studies have been made based on the Power-of-Two policy, including
Lee and Yao \cite{LY2003}, Muckstadt and Roundy \cite{MR1987}, Teo and
Bertsimas \cite{TB2001}. Teo and Bertsimas have also noted in their paper
that finding the optimal lot sizing policies for stationary demand lot
sizing problems is still an open issue.

Lu and Posner \cite{LM1994} presented a fully polynomial time approximation
scheme (FPTAS) for the Strict PJRP model with fixed base. Later, Segev \cite%
{S20012} presented a quasi-polynomial-time approximation scheme (QPTAS),
which shows that the problem is most likely not $\mathcal{APX}$-hard. In
addition, efficient polynomial time approximation scheme (EPTAS) for JRP
with finite time horizon and stationary demand was presented by Nonner and
Sviridenco \cite{NS2013}.

This problem was researched in many other different setups, such as JRP
under resource constraints (Goyal \cite{G1975}, Khouja et al \cite{K2000}
and Moon and Cha \cite{MC2006}), minimum order quantities (Porras and Dekker 
\cite{PD2006a}) and non-stationary holding cost (Levi et al. \cite{LRS2006},
Nonner and Souza \cite{NS2009}, Levi et al. \cite{L2008}).

\textbf{General PJRP.} Porras and Dekker \cite{PD2005} pointed out that
adding the correction factor leads to a completely different problem, at
least in terms of exact solvability. Porras and Dekker \cite{PD2004} show
that changing the model from Strict PJRP to PJRP significantly changes the
joint replenishment cycles and the commodities replenishment cycles. The
difference in solvability is evidenced by the sheer number of decision
variables. In the Strict PJRP all commodities cycle times are simple
functions of the joint replenishment cycle time. Thus there is actually only
a single decision variable. However, this is not the case with the PJRP
where we have $n$ decision variables, one for each commodity. We believe
this to be the main reason for the difference in the amount of research
conducted on Strict PJRP with respect to the PJRP despite the PJRP\ being
more practical. In practice, Strict PJRP is much less common than PJRP as it
involves paying for empty deliveries. Strict PJRP may occur only if there is
a binding contract with a delivery company. Although such a binding contract
may decrease the cost of the joint replenishment significantly, it usually
limits the flexibility of choosing the joint replenishment cycles. Lately,
Schulz and Telha \cite{ST2011} presented a polynomial time approximation
scheme (PTAS) for the PJRP case.

\textbf{Finite horizon.} Several heuristics were designed to deal with the
finite horizon model. Most of the finite time heuristics assume variable
demands and run-in time $\Omega \left( T\right) $\textbf{\ }\cite{LRS2006},%
\cite{J90}. Schulz and Telha \cite{ST2011} presented a polynomial-time $%
\sqrt{9/8}$-approximation algorithm for the JRP with dynamic policies and
finite horizon. As the time horizon $T$ increases, the ratio converges to $%
\sqrt{9/8}$. Schulz and Telha \cite{ST2011} also presented an FPTAS for the
Strict PJRP case with no fixed base and a finite time horizon.

Our paper proceeds as follows: In Section \ref{Model Formulation} we
formulate the problem. In Section \ref{section - infinite} we prove that the
infinite horizon PJRP is strongly $\mathcal{NP}$-hard. In Section \ref%
{section - finite} we show why the finite horizon PJRP is $\mathcal{NP}$%
-hard (not necessarily in the strong sense). Section \ref{section - Summary}
Summarizes the paper and discusses other related open problems.

\section{\label{Model Formulation} Model Formulation}

In this research, we consider the case of an infinite time horizon, and a
system composed of several commodities, for each of which there is an
external stationary demand. The demand has to be satisfied in each period.
Backlogging and lost sales are not allowed. Each commodity incurs a fixed
ordering cost for each period in which an order of the commodity is placed,
as well as a linear inventory holding cost for each period a unit of
commodity remains in storage. In addition, a joint ordering cost is incurred
for each time period where one or more orders are placed. We use the
following notations, where the units are given in square brackets:%
\begin{eqnarray*}
&&N-\text{ Number of commodities in the system }\left[ units\right] \text{,}
\\
&&\lambda _{c}-\text{Demand rate for commodity }c\text{ per period. }\left[ 
\frac{units}{period}\right] \text{,} \\
&&h_{c}-\text{Holding cost for commodity }c\text{ per period. }\left[ \frac{%
\$}{units\cdot period}\right] \text{,} \\
&&K_{c}-\text{Fixed ordering cost for commodity }c\text{ }\left[ \$\right] 
\text{,} \\
&&q_{c}-\text{ Order quantity for commodity }c\text{ }\left[ units\right] 
\text{,} \\
&&K_{0}-\text{Fixed joint ordering cost }\left[ \$\right] \text{.}
\end{eqnarray*}%
The objective is to find an integer ordering cycle time, $t_{c}$, for each
commodity $c$ so as to minimize the periodic sum of ordering and holding
costs of all commodities.

The simple model, in which there is only 1 commodity, is known as the
Economic Order Quantity (EOQ). While examining commodity $c,$ we define its 
\emph{standalone} \emph{problem} as the optimal ordering quantity problem
for a single commodity $c$ with no joint setup cost and an infinite horizon.
The standalone problem is a simple EOQ problem.

The EOQ model assumes without loss of generality that the on hand inventory
at time zero is zero. Shortage is not allowed, so we must place an order at
time zero. The average periodic cost, as a function of the cycle time $%
t_{c}, $ denoted by $g\left( t_{c}\right) $ is given by 
\begin{equation}
g\left( t_{c}\right) =\frac{K_{c}}{t_{c}}+\lambda _{c}h_{c}\frac{t_{c}}{2},
\label{EOQ formula}
\end{equation}%
and the optimal cycle time for $g\left( t_{c}\right) $, denoted by $%
t_{c}^{\ast }$ is%
\begin{equation}
t_{c}^{\ast }=\sqrt{\frac{2K_{c}}{h_{c}\lambda _{c}}}.
\label{EOQ optima solution}
\end{equation}%
In addition, $g\left( t_{c}\right) $ where $t_{c}=\beta t_{c}^{\ast }$\ (for
an arbitrary constant $\beta $) could also be calculated using:%
\begin{equation*}
g\left( \beta t_{c}^{\ast }\right) =\frac{1}{2}\left( \frac{1}{\beta }+\beta
\right) g\left( t_{c}^{\ast }\right) .
\end{equation*}

Accordingly, when debating between two options for a cycle time $\underline{t%
}_{c}$ and $\overline{t}_{c},$ such that $\underline{t}_{c}<t_{c}^{\ast }<%
\overline{t}_{c},$ then based on the standalone total cost our choice would
be:%
\begin{equation}
\begin{array}{cc}
\overline{t}_{c} & \text{if }\sqrt{\underline{t}_{c}\overline{t}_{c}}%
>t^{\ast } \\ 
\underline{t}_{c} & \text{if }\sqrt{\underline{t}_{c}\overline{t}_{c}}\leq
t^{\ast }.%
\end{array}
\label{EOQ rounding ruls}
\end{equation}%
Without loss of generality, throughout this research, we assume that $%
\lambda _{c}=2$ for all commodities.

See full elaboration and additional analysis in \cite{N2001} and \cite{Z2000}%
.

\section{\label{section - infinite}$\mathcal{NP}$-Hardness of the PJRP.}

\subsection{\label{section- reduction}A reduction from \emph{$3SAT$ }to\emph{%
\ PJRP }}

In this section we present a reduction from \emph{$3SAT$} to\ the \emph{PJRP 
}with infinite horizon.\emph{\ }The \emph{$3SAT$} is defined as follows:

\begin{definition}
Given a logical expression, $\varphi ,$ in a CNF form with $m$ clauses and $%
n $ variables, $x_{1},...,x_{n},$ where each clause, $\left( z_{i}\cup
z_{j}\cup z_{s}\right) $ where $z_{i}\in \left\{ x_{i},\overline{x}%
_{i}\right\} ,$ contains exactly 3 literals, is there a feasible assignment
to the variables such that each clause contains at least one true literal?
\end{definition}

The \emph{$3SAT$ }is strongly $\mathcal{NP}$-hard \cite{PY1988}.

In this reduction we use pairs of prime numbers with a difference of at most 
$b$ between them where $b$ is a constant. To even consider such a reduction
we have to make sure that such a set exists for any input size $n$ and that
it can be found in polynomial time. To do so we use the breakthrough proof
by Zhang \cite{Z2013}. Zhang proved that there is an infinite number of $2-$%
tuple$\left( b\right) $, prime pairs with $b\leq 7\cdot 10^{7}$, where $k-$%
tuple$\left( J\right) $ primes are a finite collection of $k=\left\vert
J\right\vert $ values representing a repeatable pattern $J$ of differences
between prime numbers. In addition, Zhang's proof could be used to show that 
$b$ is associated with another constant $\widetilde{b}$ such that there are
at least 
\begin{equation*}
\frac{x}{\log ^{\widetilde{b}}x}
\end{equation*}

$2-$tuple$\left( b\right) $ prime pairs smaller than $x.$ Since Zhang's
proof, attempts were made to decrease the bounds on both $b$ and $\widetilde{%
b}.$ The latest result, attained by the \emph{Polymath8} project \cite%
{PM2015}, sets the bounds $b\leq 256$ and $\widetilde{b}\leq 50.$

An important special case of $2-$tuple$\left( b\right) $ where $b=2$ is twin
primes. That is, twin primes are pairs of consecutive prime numbers with a
difference of exactly $2$ between them. The twin prime conjecture \cite%
{Gu1994} and the first Hardy-Littlewood conjecture \cite{HL1923} maintain
that $b=2$ and $\widetilde{b}=2.$ These values would surely make our proof
simpler. However, for sake of comprehensiveness we use general constants $b$
and $\widetilde{b}.$

In our proof we require a set of $n$ pairs of primes, denoted by $\left( 
\underline{p}_{1},\overline{p}_{1}\right), \left( \underline{p}_{2},%
\overline{p}_{2}\right) ,\ldots ,\left( \underline{p}_{n},\overline{p}%
_{n}\right) $ such that $\underline{p}_{1}<\overline{p}_{1}<\ldots <%
\underline{p}_{n}<\overline{p}_{n}$. We denote the set of primes $\{%
\underline{p}_{1},\overline{p}_{1},\ldots, \underline{p}_{n},\overline{p}%
_{n}\}$ by $\emph{VP}$ and the set of pairs $\left\{ \left( \underline{p}%
_{1},\overline{p}_{1}\right) ,\ldots ,\left( \underline{p}_{n},\overline{p}%
_{n}\right) \right\}$ by $\emph{VP}_{2}$. The primes of $\emph{VP}$ and $%
\emph{VP}_{2}$ have to satisfy the following conditions:

\begin{condition}
\label{Cond 1}The difference between the elements of a pair of consecutive
primes, denoted $b_{i}=\overline{p}_{i}-\underline{p}_{i}$ for $j=1,...,n$,
is not greater than $b.$
\end{condition}

\begin{condition}
\label{Cond 2}$\overline{p}_{n}<B\cdot \underline{p}_{1}$ where $B\geq
\left( 6\widetilde{b}\log n\right) ^{\widetilde{b}}$.
\end{condition}

\begin{condition}
\label{Cond 2.5}$\underline{p}_{1}>n^{6\widetilde{b}}.$
\end{condition}

\begin{condition}
\label{Cond 3}Any multiplication of some prime $p\in \emph{VP}$ does not
fall\ in-between any pair $\left( \underline{p}_{i},\overline{p}_{i}\right)
\in \emph{VP}_{2}.$ That is, 
\begin{equation*}
\nexists p\in \emph{VP,}\left( \underline{p}_{i},\overline{p}_{i}\right) \in 
\emph{VP}_{2},\xi \in 
\mathbb{N}
:\underline{p}_{i}<\xi \cdot p<\overline{p}_{i}.
\end{equation*}
\end{condition}

\begin{lemma}
\label{Lemma Polynomial VP}The set $\emph{VP}$ that satisfies Conditions \ref%
{Cond 1}-\ref{Cond 3} could be found in $O\left( n^{6\widetilde{b}+1}\log ^{%
\widetilde{b}}n\right) $ time.\footnote{%
See proof in the Appendix.}
\end{lemma}

Given an input of \emph{$3SAT$} problem, denoted by $\varphi $, with $n$
variables. We find a set of $n$ pairs of prime numbers that satisfy
Conditions \ref{Cond 1}-\ref{Cond 3}. We associate each pair with a variable
of $\varphi .$ The function $P\left( \cdot \right) $ is defined for each of
the variables and their negations in the CNF expression $\varphi $ as
follows: 
\begin{eqnarray*}
P\left( \overline{x}_{i}\right) &=&\underline{p}_{i} \\
P\left( x_{i}\right) &=&\overline{p}_{i}
\end{eqnarray*}

where $\left( \underline{p}_{i},\overline{p}_{i}\right) $ is the $i-th$
prime pair in $\emph{VP}_{2}$. We also define the set $\emph{PP}$ of all the
prime numbers that are smaller than $\underline{p}_{1}$. In other words, $%
\emph{PP=}\left\{ p:p<\underline{p}_{1},p\text{ is prime}\right\} $. Let us
segment the time horizon into $3$ intervals as showed in Figure \ref{F1}.
The first segment, denoted by \emph{P,} covers the interval $P=\left[ 0,%
\underline{p}_{1}\right) $. The second segment, denoted by \emph{V}, covers
the interval $V=\left[ \underline{p}_{1},\overline{p}_{n}\right] $. The last
segment, denoted by $\emph{R}$, covers the interval $R=\left( \overline{p}%
_{n},\infty \right) $. Note that \emph{PP}$\in $\emph{P} and \emph{VP}$\in $%
\emph{V}$\emph{.}$\emph{\FRAME{ftbpFU}{4.5273in}{1.2081in}{0pt}{\Qcb{Time
horizon segmentation}}{\Qlb{F1}}{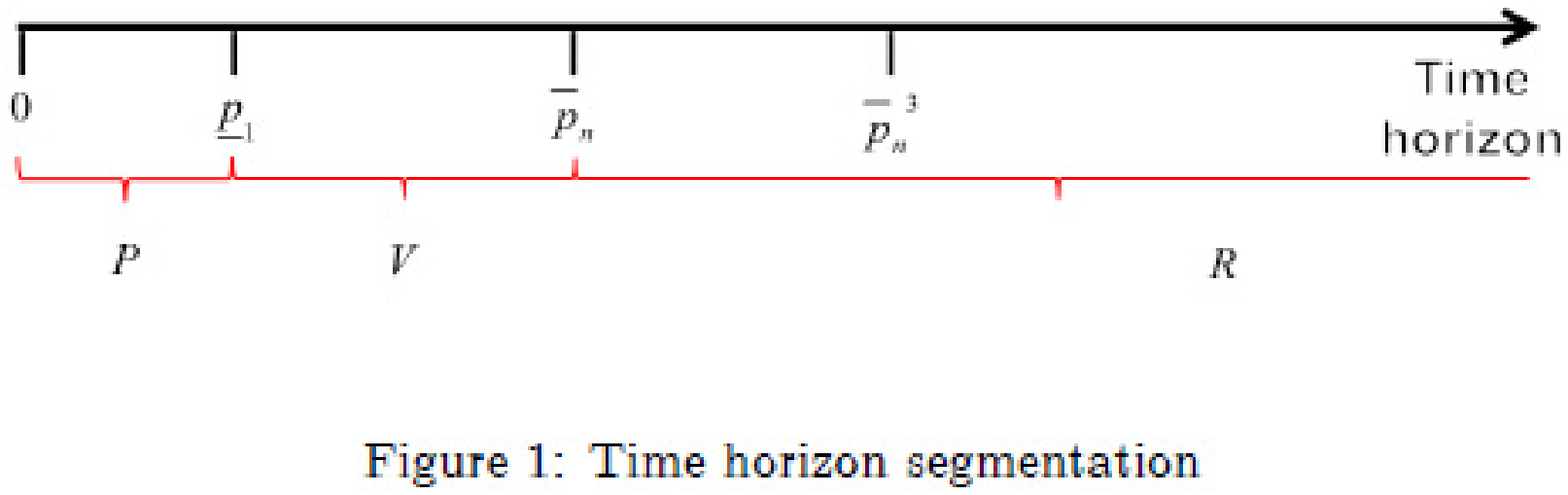}{\special{language
"Scientific Word";type "GRAPHIC";maintain-aspect-ratio TRUE;display
"USEDEF";valid_file "F";width 4.5273in;height 1.2081in;depth
0pt;original-width 17.9656in;original-height 4.7625in;cropleft "0";croptop
"1";cropright "1";cropbottom "0";filename 'time_horizon.eps';file-properties
"XNPEU";}}}

For convenience reasons we define the following quantities:%
\begin{eqnarray}
\alpha _{c} &=&\dprod\nolimits_{p\in \text{\emph{PP}}}\left( \frac{p-1}{p}%
\right)  \label{ac} \\
\alpha _{v} &=&\dprod\nolimits_{p\in \emph{VP}}\left( \frac{p-1}{p}\right)
\label{av} \\
\overline{\alpha}_{v} &=&\dprod\nolimits_{c_{i}^{x}\in \emph{Variable}%
}\left( \frac{\overline{p}_{i}-1}{\overline{p}_{i}}\right)  \label{a-v} \\
\underline{\alpha }_{v} &=&\dprod\nolimits_{c_{i}^{x}\in \emph{Variable}%
}\left( \frac{\underline{p}_{i}-1}{\underline{p}_{i}}\right)  \label{a_v} \\
a_{n} &=&\dprod {}_{\substack{ p_{\left[ j\right] }\in \emph{PP}  \\ j<n}}%
\left( 1-\frac{1}{p_{\left[ j\right] }}\right),  \label{an}
\end{eqnarray}%
where $p_{\left[ j\right]}$ is the $j^{th}$ largest prime number.

According to the time horizon segmentation we define\ the \emph{PJRP }%
instance, denoted by $\Gamma $ with $3$ sets of commodities:

\begin{itemize}
\item The first set, denoted by $\emph{Constants}$, contains commodities
with costs constructed such that their optimal cycle time is identical to
their standalone optimal cycle time, regardless of the cycle time of the
other commodities. The set $\emph{Constants}$\emph{\ }contains commodities
of the form $c_{lm}^{pv}$ for each combination of $p_{l}\in \emph{PP}$ and $%
v_{m}\in \emph{VP}$. The standalone optimal cycle time for a commodity $%
c_{lm}^{pv}$ is 
\begin{equation*}
t_{c_{lm}^{pv}}^{\ast }=p_{l}\cdot v_{m}.
\end{equation*}

The holding cost ($h_{c_{lm}^{pv}}$) and ordering cost ($K_{c_{lm}^{pv}}$)
for each commodity $c_{lm}^{pv}\in \emph{Constants}$ are as follows:%
\begin{eqnarray}
h_{c_{lm}^{pv}} &=&1  \label{DEF const h} \\
K_{c_{lm}^{pv}} &=&\left( t_{c_{lm}^{pv}}^{\ast }\right) ^{2}-\frac{1}{2}.
\label{DEF const K}
\end{eqnarray}

\item The second set, denoted by $\emph{Variables}$\emph{, }contains a
commodity $c_{i}^{x}$ for each variable $x_{i}$. We set the costs so that in
any optimal solution the cycle time of each commodity corresponding to
variable $x_{i}$ is either $\underline{p}_{i}$ or $\overline{p}_{i}$.

The holding cost ($h_{c_{i}^{x}}$) and ordering cost ($K_{c_{i}^{x}}$) for
each commodity $c_{i}^{x}\in $ $\emph{Variables}$ are as follows:%
\begin{eqnarray}
h_{c_{i}^{x}} &=&\alpha _{c}\frac{\underline{p}_{i}^{2}-b_{i}^{2}}{%
\underline{p}_{i}\left( \underline{p}_{i}+\frac{b_{i}}{2}\right) \frac{b_{i}%
}{2}}  \label{DEF vars h} \\
K_{c_{i}^{x}} &=&h_{c_{i}^{x}}\cdot \underline{p}_{i}\left( \underline{p}%
_{i}+b_{i}\right) -\frac{\underline{p}_{i}+b_{i}}{\underline{p}_{i}+b_{i}-1}%
\alpha _{c}\overline{\alpha }_{v}.  \label{DEF vars K}
\end{eqnarray}

\item The third set, denoted by \emph{Clauses}, contains a commodity $%
c_{r}^{\omega }$ for each clause $\omega _{r}=\left( z_{i}\cup z_{j}\cup
z_{s}\right) $. The standalone optimal cycle time for a commodity $%
c_{r}^{\omega }$ is 
\begin{equation}
t_{c_{r}^{\omega }}^{\ast }=P\left( z_{i}\right) \cdot P\left( z_{j}\right)
\cdot P\left( z_{s}\right) .  \label{DEF clauses t*}
\end{equation}

The holding cost ($h_{c_{r}^{\omega }}$) and ordering cost ($%
K_{c_{r}^{\omega }}$) for each commodity $c_{r}^{\omega }\in $\emph{Clauses}
are as follows:%
\begin{eqnarray}
h_{c_{r}^{\omega }} &=&1  \label{DEF clauses h} \\
K_{c_{r}^{\omega }} &=&\left( t_{c_{r}^{\omega }}^{\ast }\right) ^{2}-\frac{1%
}{2}.  \label{DEF clauses K}
\end{eqnarray}
\end{itemize}

We set the joint ordering cost to be: 
\begin{equation}
K_{0}=1.  \label{DEF K_0}
\end{equation}

\subsection{\label{section - optimality}Optimality analysis}

In this section we analyze the characteristics of the optimal solution to $%
\Gamma $. Throughout the remainder of the manuscript we use sensitivity
analysis to determine the optimality of certain cycle times. Due to the
convex nature of the cost function in Eq. (\ref{EOQ formula}) and the
discrete nature of our model, in many of our proofs it is sufficient to use
sensitivity analysis on cycle times that are within $\pm 1$ of the optimal
standalone solution. To simplify our analysis, we define the function $%
\Delta _{c}\left( t_{c},S\right) $ that describes the marginal average
periodic cost associated with commodity $c$'s cycle time, $t_{c}$, and a
solution $S$ to the other commodities in the system. We denote the lower and
upper bounds on $\Delta _{c}\left( t_{c},S\right) $ as $LB\left( \Delta
_{c}\left( t_{c},S\right) \right) $ and $UB\left( \Delta _{c}\left(
t_{c},S\right) \right) $, respectively. We also define $LB\left( \Delta
_{c}\left( t_{c}\right) \right) $ and $UB\left( \Delta _{c}\left(
t_{c}\right) \right) $ as the lower and upper bounds on the marginal average
periodic cost associated with any solution $S$ to the other commodities in
the system and with commodity $c$'s cycle time, $t_{c}$.

In the next subsections we prove that solving $\Gamma $ optimally is
equivalent to solving $\varphi $. In Section \ref{SEC cycle time constants}
we show that the cycle times of the commodities in \emph{Constants} and 
\emph{Clauses}\ are independent of the cycle times of any other commodity in
the problem. In Section \ref{SEC cycle time variables} we show that in any
optimal solution, the cycle time of each commodity $c_{i}^{x}$ $\in $ \emph{%
Variables} is either $\underline{p}_{i}$ or $\overline{p}_{i}.$ A selection
of a cycle time $\underline{p}_{i}$ or $\overline{p}_{i}$ for commodity $%
c_{i}^{x}$ $\in $\emph{Variables} is associated with assigning variable $%
x_{i}$ to either be $false$ or $true$ in $\varphi $, respectively. In
Section \ref{SEC cycle time clauses} we finalize the proof that solving $%
\Gamma $ optimally is equivalent to solving $\varphi $ by showing that in an
optimal solution the cycle times of the commodities in \emph{Variables }%
defines a solution to $\varphi $ if there is one .

\subsubsection{\label{SEC cycle time constants}Cycle time of commodities of
types \emph{Constants }and \emph{Clauses}}

In this section we show that for each commodity $c_{lm}^{pv}$ $\in $ \emph{%
Constants }and for each commodity $c_{r}^{\omega }\in $\emph{Clauses }the
cycle time in an optimal solution is $t_{c_{lm}^{pv}}^{\ast }$ and $%
t_{c_{r}^{\omega }}^{\ast }$, respectively, regardless of the cycle times of
any other commodity in the problem.

For each commodity $c_{lm}^{pv}$ $\in $ \emph{Constants }we define 2 EOQ
problems. In the first EOQ problem, denoted $\theta _{1},$ we define: $%
h_{1}=h_{c_{lm}^{pv}}$ and $K_{1}=K_{c_{lm}^{pv}}$. The solution for this
problem defines a lower bound on the marginal average periodic cost of
commodity $c_{lm}^{pv}$ assuming no joint order costs are necessary.

\begin{lemma}
\label{Lemma consts- LB}The integer optimal solution to $\theta _{1}$ is $%
t_{c_{lm}^{pv}}^{\ast }$.
\end{lemma}

\begin{proof}
According to Eq. $\left( \ref{EOQ optima solution}\right) $, the optimal
solution to the continuous $\theta _{1}$ problem, denoted $t_{1}^{\ast },$
is:%
\begin{equation}
t_{1}^{\ast }=\sqrt{\frac{2K_{^{1}}}{2h_{^{1}}}}=\sqrt{\frac{2K_{c_{lm}^{pv}}%
}{2h_{c_{lm}^{pv}}}}=\sqrt{\frac{K_{c_{lm}^{pv}}}{h_{c_{lm}^{pv}}}}.
\label{EQ const- lower bound formula}
\end{equation}%
Substituting for $h_{c_{lm}^{pv}}$ and $K_{c_{lm}^{pv}}$ using Eqs. $\left( %
\ref{DEF const h}\right) $ and $\left( \ref{DEF const K}\right) $ into Eq. $%
\left( \ref{EQ const- lower bound formula}\right) $, we get:%
\begin{equation*}
t_{1}^{\ast }=\sqrt{\frac{K_{c_{lm}^{pv}}}{h_{c_{lm}^{pv}}}}=\sqrt{\left(
t_{c_{lm}^{pv}}^{\ast }\right) ^{2}-\frac{1}{2}}
\end{equation*}%
since $\sqrt{\left( t_{c_{lm}^{pv}}^{\ast }\right) ^{2}-\frac{1}{2}}$ is not
an integer for any $t_{c_{lm}^{pv}}^{\ast }$ that is an integer, the optimal
solution will be defined according to the rounding rules in Eq. $\left( \ref%
{EOQ rounding ruls}\right) .$%
\begin{equation*}
t_{c_{lm}^{pv}}^{\ast }-1<\sqrt{\left( t_{c_{lm}^{pv}}^{\ast }\right) ^{2}-%
\frac{1}{2}}<t_{c_{lm}^{pv}}^{\ast }.
\end{equation*}%
Since%
\begin{equation*}
\sqrt{\left( t_{c_{lm}^{pv}}^{\ast }-1\right) t_{c_{lm}^{pv}}^{\ast }}=\sqrt{%
\left( t_{c_{lm}^{pv}}^{\ast }\right) ^{2}-t_{c_{lm}^{pv}}^{\ast }}<\sqrt{%
\left( t_{c_{lm}^{pv}}^{\ast }\right) ^{2}-\frac{1}{2}}
\end{equation*}%
the rounding rule states that the optimal solution will be $\left\lceil 
\sqrt{\left( t_{c_{lm}^{pv}}^{\ast }\right) ^{2}-\frac{1}{2}}\right\rceil
=t_{c_{lm}^{pv}}^{\ast }.$
\end{proof}

In the second EOQ problem, denoted $\theta_{2},$ we define: $%
h_{2}=h_{c_{lm}^{pv}}$ and $K_{2}=K_{c_{lm}^{pv}}+K_{0}$. That is, we pay $%
K_{0}$ for each order of commodity $c_{lm}^{pv}$. The solution for this
problem defines an upper bound on the marginal average periodic cost of
commodity $c_{lm}^{pv}$.

\begin{lemma}
\label{Lemma consts- UB}The integer optimal solution to $\theta_{2}$ is $%
t_{c_{lm}^{pv}}^{\ast }$.
\end{lemma}

\begin{proof}
According to Eq. $\left( \ref{EOQ optima solution}\right) $ the optimal
solution to the continuous $\theta _{2}$ problem, denoted $t_{2}^{\ast },$
is: 
\begin{equation}
t_{2}^{\ast }=\sqrt{\frac{2K_{2}}{2h_{2}}}=\sqrt{\frac{%
2K_{c_{lm}^{pv}}+2K_{0}}{2h_{c_{lm}^{pv}}}}=\sqrt{\frac{K_{c_{lm}^{pv}}+K_{0}%
}{h_{c_{lm}^{pv}}}}.  \label{EQ const- upper bound formula}
\end{equation}%
Substituting for $h_{c_{lm}^{pv}}$, $K_{c_{lm}^{pv}}$ and $K_{0}$ using Eqs. 
$\left( \ref{DEF const h}\right) ,\left( \ref{DEF const K}\right) $, and $%
\left( \ref{DEF K_0}\right) $ into Eq. $\left( \ref{EQ const- upper bound
formula}\right) $ we get:%
\begin{equation*}
t_{2}^{\ast }=\sqrt{\frac{K_{c_{lm}^{pv}}+K_{0}}{h_{c_{lm}^{pv}}}}=\sqrt{%
\left( t_{pv_{l,m}}^{\ast }\right) ^{2}+\frac{1}{2}}
\end{equation*}%
since $\sqrt{\left( t_{pv_{l,m}}^{\ast }\right) ^{2}+\frac{1}{2}}$ is not an
integer for any $t_{pv_{l,m}}^{\ast }$ that is an integer, the optimal
solution will be defined according to the rounding rules in Eq. $\left( \ref%
{EOQ rounding ruls}\right) .$%
\begin{eqnarray*}
t_{c_{lm}^{pv}}^{\ast } &<&\sqrt{\left( t_{c_{lm}^{pv}}^{\ast }\right) ^{2}+%
\frac{1}{2}}<t_{c_{lm}^{pv}}^{\ast }+1 \\
\sqrt{\left( t_{c_{lm}^{pv}}^{\ast }+1\right) t_{c_{lm}^{pv}}^{\ast }} &=&%
\sqrt{\left( t_{c_{lm}^{pv}}^{\ast }\right) ^{2}+t_{c_{lm}^{pv}}^{\ast }}>%
\sqrt{\left( t_{c_{lm}^{pv}}^{\ast }\right) ^{2}+\frac{1}{2}}
\end{eqnarray*}%
therefore according to the rounding rule, the optimal solution will be $%
\left\lfloor \sqrt{\left( t_{c_{lm}^{pv}}^{\ast }\right) ^{2}+\frac{1}{2}}%
\right\rfloor =t_{c_{lm}^{pv}}^{\ast }.$
\end{proof}

\begin{theorem}
\label{Theorem constants cycle time}In any optimal solution to $\Gamma ,$ we
have\emph{\ }$t_{c_{lm}^{pv}}=t_{c_{lm}^{pv}}^{\ast }$ for any $%
c_{lm}^{pv}\in $\emph{Constants.}
\end{theorem}

\begin{proof}
According to Lemmas \ref{Lemma consts- LB} and \ref{Lemma consts- UB} the
solutions of $P_{1}$ and $P_{2}$ that define lower and upper bounds on $%
\Delta _{c_{lm}^{pv}}\left( t_{c_{lm}^{pv}},S\right) $, respectively, are
identical. Therefore, in any optimal solution the cycle time of commodity $%
c_{lm}^{pv}$ is $t_{c_{lm}^{pv}}^{\ast }.$
\end{proof}

Using Theorem \ref{Theorem constants cycle time} we can also learn about the
cycle time of commodity $c_{r}^{\omega }\in $\emph{Clauses }in an optimal
solution to $\Gamma.$

\begin{theorem}
\label{Theorem clauses cycle time}In any optimal solution to $\Gamma ,$ we
have\emph{\ }$t_{c_{r}^{\omega }}=t_{c_{r}^{\omega }}^{\ast }$ for any $%
c_{r}^{\omega }\in $\emph{Clauses.}
\end{theorem}

\begin{proof}
Since the costs functions of commodity $c_{r}^{\omega }\in $\emph{Clauses }%
in Eqs. $\left( \ref{DEF clauses h}\right) $ and $\left( \ref{DEF clauses K}%
\right) $ are identical to the costs functions of commodity $c_{lm}^{pv}\in $%
\emph{Constants }in Eqs. $\left( \ref{DEF const h}\right) $ and $\left( \ref%
{DEF const K}\right) $, Theorem \ref{Theorem constants cycle time} holds for 
$c_{r}^{\omega }\in $\emph{Clauses }as well.
\end{proof}

\subsubsection{\label{SEC cycle time variables}Cycle time of commodities of
type \emph{Variables}}

In this section we show that for each commodity $c_{i}^{x}$ in \emph{%
Variables }the cycle time in an optimal solution is either $\underline{p}%
_{i} $ or $\overline{p}_{i}$. We denote by $jr\left( t_{c_{i}^{x}},S\right) $
the proportion of periods in which there is an order only of commodity $%
t_{c_{i}^{x}}$. Therefore, $\Delta _{c_{i}^{x}}\left( t_{c_{i}^{x}},S\right) 
$ is given by 
\begin{equation*}
\Delta _{c_{i}^{x}}\left( t_{c_{i}^{x}},S\right) =\frac{K_{c_{i}^{x}}}{%
t_{c_{i}^{x}}}+t_{c_{i}^{x}}h_{c_{i}^{x}}+K_{0}\cdot jr\left(
t_{c_{i}^{x}},S\right) .
\end{equation*}%
In order to analyze $\Delta _{c_{i}^{x}}\left( t_{c_{i}^{x}},S\right) $ we
bound $jr\left( t_{c_{i}^{x}},S\right) $. To do so we have to meticulously
calculate the average periodic marginal addition of joint replenishment cost
when choosing a cycle time of $t_{c_{i}^{x}}.$ As we shall show next, the
values in Eq. $\left( \ref{ac}\right) -\left( \ref{an}\right) $ are
meaningful and where not chosen arbitrarily.

For any prime number $p,$ the proportion of periods that are not\ a
multiplication of $p$ is 
\begin{equation*}
1-\frac{1}{p}=\frac{p-1}{p}.
\end{equation*}%
For any set of prime numbers $A$, the proportion of periods that are not
multiples of any prime number $p\in A$ is 
\begin{equation*}
\dprod\nolimits_{p\in A}\left( \frac{p-1}{p}\right) .
\end{equation*}%
Therefore, the proportion of periods that are not multiples of any prime
number $p_{l}\in \emph{PP}$ is given by:%
\begin{equation*}
\dprod\nolimits_{p\in \text{\emph{PP}}}\left( \frac{p-1}{p}\right) =\alpha
_{c}.
\end{equation*}%
Similarly, $\alpha _{v},$ $\overline{\alpha }_{v}$, $\underline{\alpha }_{v}$
and $\alpha _{n}$ represent the proportion of periods that are not multiples
of any prime number $p\in \emph{VP,}$ $\overline{p}_{i}:c_{i}^{x}\in \emph{%
Variable}$, $\underline{p}_{i}:c_{i}^{x}\in \emph{Variable}$ and $p_{\left[ i%
\right] }:i<n$, respectively.

When calculating $jr\left( t_{c},S\right) $ we may first divide the time
horizon by $t_{c}$, ending up with a new time horizon that represents only
periods where $c$ was actually ordered. Out of this new set of time periods,
denoted by $T_{t_{c}}$, we try to account for the proportion of periods
where $c$ was ordered alone and actually initiated a joint replenishment
that would not have been initiated otherwise. In other words, we are looking
for the proportion of periods not covered by other orders within $T_{t_{c}}.$
The frequency of ordering any two commodities at the same period is actually
the least common denominator of each of their individual frequencies. For
example, if commodities $c_{1}$ and $c_{2}$ have cycle times of $15$ and $9$
periods, respectively, they would be jointly ordered every $45$ periods. If
we were to consider $T_{t_{c_{1}}}$, which divides the time horizon by $%
t_{1}=15,$ then $1/3$ of the periods in $T_{t_{c_{1}}}$ would have already
been covered by $c_{2}$.

According to Theorem \ref{Theorem constants cycle time}, in any optimal
solution to $\Gamma $\emph{\ }$t_{c_{i}^{x}}=t_{c_{i}^{x}}^{\ast }$ for any $%
c_{i}^{x}\in $\emph{Constants}, and therefore we can consider only solutions
in which $t_{c_{i}^{x}}=t_{c_{i}^{x}}^{\ast }$ for any $c_{i}^{x}\in $\emph{%
Constants}. Note that since there is a commodity $c_{i}^{x}\in $\emph{%
Constants }for each combination of $p_{l}\in \emph{PP}$ and $v_{m}\in \emph{%
VP}$, we can assume that at any period that is a multiplication of $%
\underline{p}_{i}$ and $p_{l}\in \emph{PP}$ or a multiplication of $%
\overline{p}_{i}$ and $p_{l}\in \emph{PP}$, there is an order placed due to
the commodities in \emph{Constants}. Therefore $jr\left(
t_{c_{i}^{x}},S\right) $ is not greater than the proportion of periods whose
factors include $t_{c_{i}^{x}}\ $and exclude all prime numbers of set $\emph{%
PP}$. This proportion is given by: 
\begin{equation*}
jr\left( t_{c_{i}^{x}},S\right) <\frac{1}{t_{c_{i}^{x}}}\cdot \alpha _{c}.
\end{equation*}%
Therefore, 
\begin{equation}
UB\left( \Delta _{c_{i}^{x}}\left( t_{c_{i}^{x}}\in \left\{ \underline{p}%
_{i},\overline{p}_{i}\right\} ,S\right) \right) =\frac{K_{c_{i}^{x}}}{%
t_{c_{i}^{x}}}+t_{c_{i}^{x}}h_{c_{i}^{x}}+K_{0}\cdot \frac{1}{t_{c_{i}^{x}}}%
\cdot \alpha _{c}.  \label{EQ vars UB p- or p+}
\end{equation}

The lower bound on this marginal cost for an arbitrary cycle time $%
t_{c_{i}^{x}}$ is given by the solution $S$ in which there is a cycle time $%
t\in S$ such that $\func{mod}(t_{i},t)=0$. In this case $jr\left(
t_{c_{i}^{x}},S\right) =0$ and then the lower bound is given by:%
\begin{equation}
LB\left( \Delta _{c_{i}^{x}}\left( t_{c_{i}^{x}}\right) \right) =\frac{%
K_{c_{i}^{x}}}{t_{c_{i}^{x}}}+t_{c_{i}^{x}}h_{c_{i}^{x}}.
\label{EQ vars LB t!= p- or p+}
\end{equation}%
In order to prove that for each commodity $c_{i}^{x}\in $\emph{Variables }%
the cycle time in an optimal solution is either $\underline{p}_{i}$ or $%
\overline{p}_{i}$, we first show that the optimal solution is bounded by the
range $\left[ \underline{p}_{i},\overline{p}_{i}\right] $.

\begin{claim}
\label{claim vars A p+<p+ +1}For each $c_{i}^{x}$ $\in $\emph{Variables} and
for every solution $S$, $\Delta _{c_{i}^{x}}\left( \overline{p}_{i},S\right)
\leq \Delta _{c_{i}^{x}}\left( \overline{p}_{i}+1,S\right) $
\end{claim}

\begin{proof}
We show now that $UB\left( \Delta _{c_{i}^{x}}\left( \overline{p}_{i}\right)
\right) \leq LB\left( \Delta _{c_{i}^{x}}\left( \overline{p}_{i}+1\right)
\right) .$ Using Eqs. $\left( \ref{EQ vars UB p- or p+}\right) $ and $\left( %
\ref{EQ vars LB t!= p- or p+}\right) $ with $t_{c_{i}^{x}}=\overline{p}_{i}$
and $t_{c_{i}^{x}}=\overline{p}_{i}+1$, respectively, we get:%
\begin{eqnarray*}
UB\left( \Delta _{c_{i}^{x}}\left( \overline{p}_{i}\right) \right) &=&\frac{%
K_{c_{i}^{x}}}{\overline{p}_{i}}+\overline{p}_{i}h_{c_{i}^{x}}+K_{0}\cdot 
\frac{1}{\overline{p}_{i}}\cdot \alpha _{c} \\
LB\left( \Delta _{c_{i}^{x}}\left( \overline{p}_{i}+1\right) \right) &=&%
\frac{K_{c_{i}^{x}}}{\left( \overline{p}_{i}+1\right) }+\left( \overline{p}%
_{i}+1\right) h_{c_{i}^{x}}.
\end{eqnarray*}%
Therefore, 
\begin{eqnarray}
&&LB\left( \Delta _{c_{i}^{x}}\left( \overline{p}_{i}+1\right) \right)
-UB\left( \Delta _{c_{i}^{x}}\left( \overline{p}_{i}\right) \right)  \notag
\\
&=&\frac{K_{c_{i}^{x}}}{\left( \overline{p}_{i}+1\right) }+\left( \overline{p%
}_{i}+1\right) h_{c_{i}^{x}}-\frac{K_{c_{i}^{x}}}{\overline{p}_{i}}-%
\overline{p}_{i}h_{c_{i}^{x}}-K_{0}\cdot \frac{1}{\overline{p}_{i}}\cdot
\alpha _{c}  \notag \\
&=&h_{c_{i}^{x}}-\frac{K_{c_{i}^{x}}}{\overline{p}_{i}\left( \overline{p}%
_{i}+1\right) }-K_{0}\cdot \frac{1}{\overline{p}_{i}}\cdot \alpha _{c}.
\label{EQ LB(p+ +1) -UB(p+) with p+}
\end{eqnarray}%
Substituting for $\overline{p}_{i}=\underline{p}_{i}+b_{i}$ (see Condition %
\ref{Cond 1}) into Eq. $\left( \ref{EQ LB(p+ +1) -UB(p+) with p+}\right) $
we get:%
\begin{eqnarray}
&&LB\left( \Delta _{c_{i}^{x}}\left( \overline{p}_{i}+1\right) \right)
-UB\left( \Delta _{c_{i}^{x}}\left( \overline{p}_{i}\right) \right)  \notag
\\
&=&h_{c_{i}^{x}}-\frac{K_{c_{i}^{x}}}{\left( \underline{p}_{i}+b_{i}\right)
\left( \underline{p}_{i}+b_{i}+1\right) }-K_{0}\cdot \frac{1}{\left( 
\underline{p}_{i}+b_{i}\right) }\cdot \alpha _{c}.
\label{EQ LB(p+ +1) -UB(p+) before h and k}
\end{eqnarray}%
Substituting for $h_{c_{i}^{x}},$ $K_{c_{i}^{x}}$, and $K_{0}$ using Eqs. $%
\left( \ref{DEF vars h}\right) $, $\left( \ref{DEF vars K}\right) $, and $%
\left( \ref{DEF K_0}\right) $ into Eq. $\left( \ref{EQ LB(p+ +1) -UB(p+)
before h and k}\right) $ we get:%
\begin{eqnarray*}
&&LB\left( \Delta _{c_{i}^{x}}\left( \overline{p}_{i}+1\right) \right)
-UB\left( \Delta _{c_{i}^{x}}\left( \overline{p}_{i}\right) \right) \\
&=&h_{c_{i}^{x}}-\frac{h_{c_{i}^{x}}\cdot \underline{p}_{i}\left( \underline{%
p}_{i}+b_{i}\right) }{\left( \underline{p}_{i}+b_{i}\right) \left( 
\underline{p}_{i}+b_{i}+1\right) }+\frac{\frac{\underline{p}_{i}+b_{i}}{%
\underline{p}_{i}+b_{i}-1}\alpha _{c}\overline{\alpha }_{v}}{\left( 
\underline{p}_{i}+b_{i}\right) \left( \underline{p}_{i}+b_{i}+1\right) }-%
\frac{1}{\left( \underline{p}_{i}+b_{i}\right) }\cdot \alpha _{c} \\
&=&h_{c_{i}^{x}}\frac{b_{i}+1}{\left( \underline{p}_{i}+b_{i}+1\right) }+%
\frac{\alpha _{c}\overline{\alpha }_{v}}{\left( \underline{p}%
_{i}+b_{i}-1\right) \left( \underline{p}_{i}+b_{i}+1\right) }-\frac{1}{%
\left( \underline{p}_{i}+b_{i}\right) }\cdot \alpha _{c} \\
&=&\alpha _{c}\left( \frac{\left( \underline{p}_{i}^{2}-b_{i}^{2}\right)
\left( b_{i}+1\right) }{\underline{p}_{i}\left( \underline{p}_{i}+\frac{b_{i}%
}{2}\right) \frac{b_{i}}{2}\left( \underline{p}_{i}+b_{i}+1\right) }+\frac{%
\overline{\alpha }_{v}}{\left( \underline{p}_{i}+b_{i}-1\right) \left( 
\underline{p}_{i}+b_{i}+1\right) }-\frac{1}{\left( \underline{p}%
_{i}+b_{i}\right) }\right) \\
&>&\alpha _{c}\left( \frac{2\left( \underline{p}_{i}^{2}-b_{i}^{2}\right) }{%
\underline{p}_{i}\left( \underline{p}_{i}+\frac{b_{i}}{2}\right) \left( 
\underline{p}_{i}+b_{i}+1\right) }-\frac{1}{\left( \underline{p}_{i}+\frac{%
b_{i}}{2}\right) }\right) \\
&=&\alpha _{c}\frac{\underline{p}_{i}^{2}-2b_{i}^{2}-\underline{p}_{i}b_{i}-%
\underline{p}_{i}}{\underline{p}_{i}\left( \underline{p}_{i}+\frac{b_{i}}{2}%
\right) \left( \underline{p}_{i}+b_{i}+1\right) }.
\end{eqnarray*}%
$\underline{p}_{i}$ is bounded below by $n^{6\widetilde{b}}$ (See Condition %
\ref{Cond 2.5}); thus, for an input $n>2$ the numerator is positive even for
the upper bound of $256$ on $b$ (\cite{PM2015}) and a lower bound of $2$ on $%
\widetilde{b};$ thus, for any permissible $\underline{p}_{i}:$%
\begin{equation*}
LB\left( \Delta _{c_{i}^{x}}\left( \overline{p}_{i}+1\right) \right)
-UB\left( \Delta _{c_{i}^{x}}\left( \overline{p}_{i}\right) \right) >0
\end{equation*}
\end{proof}

\begin{claim}
\label{claim vars A p-<p- -1}For each $c_{i}^{x}$ $\in $\emph{Variables} and
for every solution $S$, $\Delta _{c_{i}^{x}}\left( \overline{p}_{i},S\right)
\leq \Delta _{c_{i}^{x}}\left( \underline{p}_{i}-1,S\right) $
\end{claim}

\begin{proof}
We show now that $LB\left( \Delta _{c_{i}^{x}}\left( \underline{p}%
_{i}+b_{i}+1\right) \right) \leq LB\left( \Delta _{c_{i}^{x}}\left( 
\underline{p}_{i}-1\right) \right) $. Note that $LB\left( \Delta
_{c_{i}^{x}}\left( \underline{p}_{i}+b_{i}+1\right) \right) $ and $LB\left(
\Delta _{c_{i}^{x}}\left( \underline{p}_{i}-1\right) \right) $ are in fact
the standalone costs $g\left( \underline{p}_{i}+b_{i}+1\right) $ and $%
g\left( \underline{p}_{i}-1\right) ,$\ respectively. According to the
rounding rules in Eq. $\left( \ref{EOQ rounding ruls}\right) :$ $g\left( 
\underline{p}_{i}+b_{i}+1\right) <g\left( \underline{p}_{i}-1\right) $ if $%
\underline{p}_{i}-1<t_{c_{i}^{x}}^{\ast }<\underline{p}_{i}+b_{i}+1$ and 
\begin{equation*}
\sqrt{\left( \underline{p}_{i}+b_{i}+1\right) \left( \underline{p}%
_{i}-1\right) }<t_{c_{i}^{x}}^{\ast }.
\end{equation*}%
The optimal solution to the standalone problem, $t_{c_{i}^{x}}^{\ast }$ is
given by: 
\begin{eqnarray*}
\sqrt{\frac{K_{c_{i}^{x}}}{h_{c_{i}^{x}}}} &=&\sqrt{\underline{p}_{i}\left( 
\underline{p}_{i}+b_{i}\right) -\frac{\frac{\underline{p}_{i}+b_{i}}{%
\underline{p}_{i}+b_{i}-1}\alpha _{c}\overline{\alpha }_{v}}{h_{c_{i}^{x}}}}=%
\sqrt{\underline{p}_{i}\left( \underline{p}_{i}+b_{i}\right) -\frac{\frac{%
\underline{p}_{i}+b_{i}}{\underline{p}_{i}+b_{i}-1}\overline{\alpha }_{v}}{%
\left( \frac{\underline{p}_{i}^{2}-b_{i}^{2}}{\underline{p}_{i}\left( 
\underline{p}_{i}+\frac{b_{i}}{2}\right) \frac{b_{i}}{2}}\right) }} \\
&>&\sqrt{\underline{p}_{i}\left( \underline{p}_{i}+b_{i}\right) -\frac{%
\underline{p}_{i}\left( \underline{p}_{i}+\frac{b_{i}}{2}\right) \frac{b_{i}%
}{2}}{\left( \underline{p}_{i}+b_{i}-1\right) \left( \underline{p}%
_{i}-b_{i}\right) }} \\
&>&\sqrt{\underline{p}_{i}\left( \underline{p}_{i}+b_{i}\right) -\frac{b_{i}%
\underline{p}_{i}}{2\left( \underline{p}_{i}-b_{i}\right) }}>\sqrt{%
\underline{p}_{i}^{2}+b_{i}\underline{p}_{i}-b_{i}} \\
&>&\sqrt{\underline{p}_{i}^{2}+b_{i}\underline{p}_{i}-b_{i}-1}=\sqrt{\left( 
\underline{p}_{i}+b_{i}+1\right) \left( \underline{p}_{i}-1\right) }.
\end{eqnarray*}%
Therefore, $LB\left( \Delta _{c_{i}^{x}}\left( \underline{p}%
_{i}+b_{i}+1\right) \right) \leq LB\left( \Delta _{c_{i}^{x}}\left( 
\underline{p}_{i}-1\right) \right) $. According to Claim \ref{claim vars A
p+<p+ +1} $\Delta _{c_{i}^{x}}\left( \overline{p}_{i},S\right) \leq \Delta
_{c_{i}^{x}}\left( \overline{p}_{i}+1,S\right) $; hence, $\Delta
_{c_{i}^{x}}\left( \overline{p}_{i},S\right) \leq \Delta _{c_{i}^{x}}\left( 
\overline{p}_{i}+1,S\right) \leq \Delta _{c_{i}^{x}}\left( \underline{p}%
_{i}-1,S\right) $.
\end{proof}

\begin{theorem}
\label{Theorem vars- bounded range of solution}In any optimal solution to $%
\Gamma ,$\emph{\ }$t_{c_{i}^{x}}\in \left[ \underline{p}_{i},\overline{p}_{i}%
\right] $ for any $c_{i}^{x}\in $ \emph{Variables.}
\end{theorem}

\begin{proof}
According to Claims $\ref{claim vars A p+<p+ +1}$ and $\ref{claim vars A
p-<p- -1}$ for each variable $c_{i}^{x}\in $\emph{Variable}, $UB\left(
\Delta _{c_{i}^{x}}\left( \overline{p}_{i}\right) \right) \leq LB\left(
\Delta _{c_{i}^{x}}\left( \overline{p}_{i}+1\right) \right) \leq LB\left(
\Delta _{c_{i}^{x}}\left( \underline{p}_{i}-1\right) \right) $. Due to the
convex nature of the cost function and since $\underline{p}%
_{i}<t_{c_{i}^{x}}^{\ast }<\overline{p}_{i}$, $LB\left( \Delta
_{c_{i}^{x}}\left( \overline{p}_{i}+1\right) \right) $ is a lower bound on
any solution $t_{c_{i}^{x}}\notin \left[ \underline{p}_{i},\overline{p}_{i}%
\right] $. Therefore in any optimal solution to $\Gamma ,$\emph{\ }$%
t_{c_{i}^{x}}\in \left[ \underline{p}_{i},\overline{p}_{i}\right] .$
\end{proof}

Accordingly, we consider only\emph{\ }$t_{c_{i}^{x}}\in \left\{ \underline{p}%
_{i}+y:0\leq y\leq b_{i}\right\}$. In the next claim we prove that any
solution \emph{\ }$t_{c_{i}^{x}}\not\in \left\{ \underline{p}_{i},\overline{p%
}_{i}\right\}$ is not the optimal solution. Let us assume that $%
t_{c_{i}^{x}}\in \left\{ \underline{p}_{i}+y:0<y<b_{i}\right\}$ and find a
new lower bound for the solution. Note that $t_{c_{i}^{x}}$ might not be a
prime number.

Let us calculate the lower bound for $jr\left( t_{c_{i}^{x}},S\right) .$
According to Condition \ref{Cond 3} none of the factors of $t_{c_{i}^{x}}$
belong to $\emph{VP}$. However, the factorials of $t_{c_{i}^{x}}$ may
include primes $p\in \emph{PP}$. If that happens, Theorem \ref{Theorem
constants cycle time} states that at each period that is a\ multiple of a
prime number $p\in \emph{VP}$ and $t_{c_{j}^{x}}$, there is an order of
another commodity $c_{lm}^{pv}$ $\in $ \emph{Constants}. Thus, as a lower
bound of $jr\left( t_{c_{i}^{x}},S\right) ,$ $\frac{1}{\underline{p}_{i}+y}%
\left( 1-\alpha _{v}\right) $ of the periods may already be covered. Of the
remaining $\frac{1}{\underline{p}_{i}+y}\alpha _{v}$ periods there might be
periods covered by some other $c_{j}^{x}$ $\in $ \emph{Variables }sharing
the same factorials with $t_{c_{i}^{x}}$. Any two commodities $%
c_{i}^{x},c_{j}^{x}$ $\in $ \emph{Variables} with $t_{c_{i}^{x}}$ and $%
t_{c_{j}^{x}}$ that are not primes, might share some common prime factors as
well as a unique multiplier (might not be a prime one) of each one. Hence,
we may represent their respective cycle times by $t_{c_{i}^{x}}=\eta \mu
_{i} $ and $t_{c_{j}^{x}}=\eta \mu _{j}$, where $\mu _{i}$ and $\mu _{j}$
represent the unique elements of $t_{c_{i}^{x}}$, and $t_{c_{j}^{x}}$ and $%
\eta $ and represent their common factors.

Calculating $jr\left( t_{c_{i}^{x}},S\right) $, each period in $%
T_{t_{c_{i}^{x}}}$ that is a\ multiple of $\mu _{j}$ is covered by $%
c_{j}^{x} $. Since there are $n$ commodities of type $c_{i}^{x}$ $\in $\emph{%
Variables,} there are at most $n-1$ such unique $\mu _{j}$ elements. The
more common factors these elements share and the smaller they are the more
time periods they will cover in $T_{t_{c_{i}^{x}}}.$ Accordingly, a lower
bound for $jr\left( \underline{p}_{i}+y,S\right) $ considers the $n$
smallest primes as $\mu _{j}$ values. 
\begin{equation*}
jr\left( t_{c_{i}^{x}},S\right) >\frac{1}{t_{c_{i}^{x}}}\left( \alpha
_{v}\dprod {}_{\substack{ q_{\left[ j\right] }\in \emph{PP}  \\ j<n}}\left(
1-\frac{1}{q_{\left[ j\right] }}\right) \right) =\frac{1}{t_{c_{i}^{x}}}%
\left( \alpha _{v}a_{n}\right) .
\end{equation*}%
Hence, assuming the optimal cycle time according to Eq. $\left( \ref{EOQ
optima solution}\right) $: 
\begin{equation}
LB\left( \Delta _{c_{i}^{x}}\left( t_{c_{i}^{x}}\notin \left\{ \underline{p}%
_{i},\overline{p}_{i}\right\} \right) \right) =\frac{K_{c_{i}^{x}}+K_{0}%
\left( \alpha _{v}a_{n}\right) }{t_{c_{i}^{x}}}+h_{c_{i}^{x}}\left(
t_{c_{i}^{x}}\right) .  \label{EQ LB M tight}
\end{equation}

In the next Lemmas we show that for each $c_{i}^{x}$ $\in A$ and for every
solution $S$ the optimal $t_{c_{i}}\ $is either $\underline{p}_{i}$ or $%
\overline{p}_{i}.$

\begin{claim}
\label{Claim vars x+ better then M}For each $c_{i}^{x}$ $\in $\emph{Variables%
} and for every solution $S$, $\Delta _{c_{i}^{x}}\left( \underline{p}%
_{i},S\right) <\Delta _{c_{i}^{x}}\left( \underline{p}_{i}+y,S\right)$ for $%
0<y<b_{i}.$
\end{claim}

\begin{proof}
We show now that $UB\left( \Delta _{c_{i}^{x}}\left( \underline{p}%
_{i}\right) \right) <LB\left( \Delta _{c_{i}^{x}}\left( \underline{p}%
_{i}+y\right) \right).$ Using Eq. $\left( \ref{EQ LB M tight}\right)$ with $%
t_{c_{i}^{x}}=\underline{p}_{i}+y$ and Eq. $\left( \ref{EQ vars UB p- or p+}%
\right)$ we get:%
\begin{eqnarray}
&&LB\left( \Delta_{c_{i}^{x}}\left( \underline{p}_{i}+y\right) \right)
-UB\left( \Delta_{c_{i}^{x}}\left( \underline{p}_{i}\right) \right)  \notag
\\
&=&\frac{K_{c_{i}^{x}}+K_{0}\left( \alpha _{v}a_{n}\right) }{\underline{p}%
_{i}+y}+h_{c_{i}^{x}}\left( \underline{p}_{i}+y\right) -\left( \frac{%
K_{c_{i}^{x}}}{\underline{p}_{i}}+\underline{p}_{i}h_{c_{i}^{x}}+K_{0}\cdot 
\frac{\alpha _{c}}{\underline{p}_{i}}\right)  \notag \\
&=&\frac{-yK_{c_{i}^{x}}}{\underline{p}_{i}\left( \underline{p}_{i}+y\right) 
}+yh_{c_{i}^{x}}+\frac{\alpha_{v}a_{n}}{\underline{p}_{i}+y}-\frac{\alpha
_{c}}{\underline{p}_{i}}  \label{EQ LB(p+ -1) -UB(p+) with p+}
\end{eqnarray}%
Substituting for $h_{c_{i}^{x}},$ $K_{c_{i}^{x}}$ and $K_{0}$ using Eqs. $%
\left( \ref{DEF vars h}\right)$, $\left( \ref{DEF vars K}\right)$ and $%
\left( \ref{DEF K_0}\right) $ into Eq. $\left( \ref{EQ LB(p+ -1) -UB(p+)
with p+}\right)$ we get:%
\begin{eqnarray*}
&&LB\left( \Delta _{c_{i}^{x}}\left( \underline{p}_{i}+y\right) \right)
-UB\left( \Delta _{c_{i}^{x}}\left( \underline{p}_{i}\right) \right) \\
&=&\frac{-yh_{c_{i}^{x}}\cdot \underline{p}_{i}\left( \underline{p}%
_{i}+b_{i}\right) +y\frac{\underline{p}_{i}+b_{i}}{\underline{p}_{i}+b_{i}-1}%
\alpha _{c}\overline{\alpha }_{v}}{\underline{p}_{i}\left( \underline{p}%
_{i}+y\right) }+yh_{c_{i}^{x}}+\frac{\alpha _{v}a_{n}}{\underline{p}_{i}+y}-%
\frac{\alpha _{c}}{\underline{p}_{i}} \\
&=&-yh_{c_{i}^{x}}\frac{b_{i}-y}{\underline{p}_{i}+y}+\frac{y\left( 
\underline{p}_{i}+b_{i}\right) \alpha_{c}\overline{\alpha }_{v}}{\underline{p%
}_{i}\left( \underline{p}_{i}+y\right) \left( \underline{p}%
_{i}+b_{i}-1\right) }+\frac{\alpha _{v}a_{n}}{\underline{p}_{i}+y}-\frac{%
\alpha _{c}}{\underline{p}_{i}} \\
&=&-\alpha _{c}y\frac{\underline{p}_{i}^{2}-b_{i}^{2}}{\underline{p}%
_{i}\left( \underline{p}_{i}+\frac{b_{i}}{2}\right) \frac{b_{i}}{2}}\frac{%
b_{i}-y}{\underline{p}_{i}+y}+\frac{y\left( \underline{p}_{i}+b_{i}\right)
\alpha _{c}\overline{\alpha }_{v}}{\underline{p}_{i}\left( \underline{p}%
_{i}+y\right) \left( \underline{p}_{i}+b_{i}-1\right) }+\frac{\alpha
_{v}a_{n}}{\underline{p}_{i}+y}-\frac{\alpha _{c}}{\underline{p}_{i}} \\
&=&\frac{-\alpha _{c}}{\underline{p}_{i}+y}\left( y\frac{\underline{p}%
_{i}^{2}-b_{i}^{2}}{\underline{p}_{i}\left( \underline{p}_{i}+\frac{b_{i}}{2}%
\right) \frac{b_{i}}{2}}\left( b_{i}-y\right) -\frac{y\left( \underline{p}%
_{i}+b_{i}\right) \overline{\alpha }_{v}}{\underline{p}_{i}\left( \underline{%
p}_{i}+b_{i}-1\right) }+\frac{\left( \underline{p}_{i}+y\right) }{\underline{%
p}_{i}}\right) +\frac{\alpha _{v}a_{n}}{\underline{p}_{i}+y} \\
&>&\frac{-\alpha _{c}}{\underline{p}_{i}+y}\left( \frac{2y\left(
b_{i}-y\right) }{b_{i}}+2\right) +\frac{\alpha _{v}a_{n}}{\underline{p}_{i}+y%
}=\frac{a_{n}}{\underline{p}_{i}+y}\left( \alpha _{v}-2\frac{\alpha _{c}}{%
a_{n}}\left( \frac{y\left( b_{i}-y\right) }{b_{i}}+1\right) \right) .
\end{eqnarray*}%
The value $y\left( b_{i}-y\right) \ $is maximized when $y=0.5b_{i};$ thus:%
\begin{equation}
LB\left( \Delta_{c_{i}^{x}}\left( \underline{p}_{i}+y\right) \right)
-UB\left( \Delta_{c_{i}^{x}}\left( \underline{p}_{i}\right) \right) >\frac{%
a_{n}}{\underline{p}_{i}+y}\left( \alpha_{v}-\frac{\alpha _{c}}{a_{n}}\left( 
\frac{b_{i}}{2}+2\right) \right)  \label{Eq. y=0.5b}
\end{equation}

According to Ribenboim \cite{R1996} and Condition \ref{Cond 2.5} there are
at least 
\begin{equation*}
\frac{0.91\cdot n^{6\widetilde{b}}}{6\widetilde{b}\log n}>3n
\end{equation*}%
prime numbers in $\emph{PP.}$ Moreover, $\alpha _{c}$ and $a_{n}$ share the
first $n$ elements of their respective multiples. Therefore, we can cancel
out these $n$ elements and explicitly write $\frac{\alpha _{c}}{a_{n}}$ in
Eq. $\left( \ref{Eq. y=0.5b}\right) $ as follows:%
\begin{eqnarray*}
&&LB\left( \Delta _{c_{i}^{x}}\left( \underline{p}_{i}+y\right) \right)
-UB\left( \Delta _{c_{i}^{x}}\left( \underline{p}_{i}\right) \right) \\
&>&\frac{a_{n}}{\underline{p}_{i}+1}\left( \alpha _{v}-\left( \frac{b_{i}}{2}%
+2\right) \prod_{_{\substack{ j\geq n  \\ p_{\left[ j\right] }\leq n^{6%
\widetilde{b}}}}}\left( \frac{p_{\left[ j\right] }-1}{p_{\left[ j\right] }}%
\right) \right) \\
&=&\frac{a_{n}}{\underline{p}_{i}+1}\left( \alpha _{v}-\left( \frac{b_{i}}{2}%
+2\right) \prod_{n\leq j\leq 3n}\left( \frac{p_{\left[ j\right] }-1}{p_{%
\left[ j\right] }}\right) \prod_{_{\substack{ j\geq 3n  \\ p_{\left[ j\right]
}\leq n^{6\widetilde{b}}}}}\left( \frac{p_{\left[ j\right] }-1}{p_{\left[ j%
\right] }}\right) \right) .
\end{eqnarray*}

Note that both $\alpha _{v}$ and $\dprod\nolimits_{n\leq j\leq 3n}\left( 
\frac{p_{\left[ j\right] }-1}{p_{\left[ j\right] }}\right) $ are multiples
of $2n$ elements where each element in $\alpha _{v}$ is bigger than each
element in $\dprod\nolimits_{n\leq j\leq 3n}\left( \frac{p_{\left[ j\right]
}-1}{p_{\left[ j\right] }}\right) $ and therefore $\alpha
_{v}>\dprod\nolimits_{n\leq j\leq 3n}\left( \frac{p_{\left[ j\right] }-1}{p_{%
\left[ j\right] }}\right) $. Hence,%
\begin{eqnarray*}
&&LB\left( \Delta _{c_{i}^{x}}\left( \underline{p}_{i}+y\right) \right)
-UB\left( \Delta _{c_{i}^{x}}\left( \underline{p}_{i}\right) \right) \\
&>&\frac{a_{n}\alpha _{v}}{\underline{p}_{i}+1}\left( 1-\left( \frac{b_{i}}{2%
}+2\right) \dprod\nolimits_{\substack{ 3n<j  \\ p_{\left[ j\right] }\leq n^{6%
\widetilde{b}}}}\left( \frac{p_{\left[ j\right] }-1}{p_{\left[ j\right] }}%
\right) \right) .
\end{eqnarray*}

Using the upper bound of $256$ on $b_{i}$ (see \cite{PM2015}) we have: 
\begin{eqnarray}
&&LB\left( \Delta_{c_{i}^{x}}\left( \underline{p}_{i}+y\right) \right)
-UB\left( \Delta_{c_{i}^{x}}\left( \underline{p}_{i}\right) \right)  \notag
\\
&>&\frac{a_{n}\alpha _{v}}{\underline{p}_{i}+1}\left( 1-130\dprod\nolimits 
_{\substack{ 3n<j  \\ p_{\left[ j\right] }\leq n^{6\widetilde{b}}}}\left( 
\frac{p_{\left[ j\right] }-1}{p_{\left[ j\right] }}\right) \right).
\label{EQ vars M- before bounding p_3n}
\end{eqnarray}

In order to lower bound $p_{\left[ 3n\right] }$ we use the bound presented
in \cite{R1996}:%
\begin{equation*}
p_{\left[ 3n\right] }>0.91\cdot 3n\ln \left( 3n\right).
\end{equation*}%
Substituting this bound into Eq. $\left( \ref{EQ vars M- before bounding
p_3n}\right)$ we get: 
\begin{eqnarray}
&&LB\left( \Delta_{c_{i}^{x}}\left( \underline{p}_{i}+y\right) \right)
-UB\left( \Delta_{c_{i}^{x}}\left( \underline{p}_{i}\right) \right)  \notag
\\
&>&\frac{a_{n}\alpha_{v}}{\underline{p}_{i}+1}\left(
1-130\dprod\nolimits_{0.91\cdot 3n\ln 3n<p_{\left[ j\right] }\leq n^{6%
\widetilde{b}}}\left( \frac{p_{\left[ j\right] }-1}{p_{\left[ j\right] }}%
\right) \right).  \label{EQ vars M- before mertens}
\end{eqnarray}

According to Merten's theorems \cite{M1874}:%
\begin{equation*}
\dprod\nolimits_{j\leq G}\left( \frac{p_{\left[ j\right] }-1}{p_{\left[ j%
\right] }}\right) =\frac{e^{-\gamma \rho \left( G\right) }}{\ln \left(
G\right) },
\end{equation*}%
where $0<\rho \left( G\right) <\frac{4}{\ln \left( G+1\right) }+\frac{2}{%
G\ln \left( G\right) }+\frac{1}{2G}$ and $\gamma $ is the Euler--Mascheroni
constant. Hence, 
\begin{equation*}
\dprod\nolimits_{G^{\prime }<j\leq G}\left( \frac{p_{\left[ j\right] }-1}{p_{%
\left[ j\right] }}\right) =\frac{\frac{e^{-\gamma \rho \left( G\right) }}{%
\ln \left( G\right) }}{\frac{e^{-\gamma \rho \left( G^{\prime }\right) }}{%
\ln \left( G^{\prime }\right) }}=\frac{\ln \left( G^{\prime }\right)
e^{\gamma \left( \rho \left( G^{\prime }\right) -\rho \left( G\right)
\right) }}{\ln \left( G\right) }.
\end{equation*}%
Therefore, the bound on $\dprod\nolimits_{0.91\cdot 3n\ln 3n<p_{\left[ j%
\right] }\leq n^{6\widetilde{b}}}\left( \frac{p_{\left[ j\right] }-1}{p_{%
\left[ j\right] }}\right) $ is given by: 
\begin{eqnarray}
&&\dprod\nolimits_{0.91\cdot 3n\ln 3n<p_{\left[ j\right] }\leq n^{6%
\widetilde{b}}}\left( \frac{p_{\left[ j\right] }-1}{p_{\left[ j\right] }}%
\right)  \notag \\
&=&\frac{\ln \left( 0.91\cdot 3n\ln 3n\right) e^{\gamma \left( \rho \left(
0.91\cdot 3n\ln 3n\right) -\rho \left( n^{6\widetilde{b}}\right) \right) }}{%
\ln \left( n^{6\widetilde{b}}\right) }  \notag \\
&<&\frac{\ln \left( 0.91\cdot 3n\ln 3n\right) }{6\widetilde{b}\ln n}%
e^{\gamma \left( \frac{4}{\ln \left( 0.91\cdot 3n\ln 3n+1\right) }+\frac{2}{%
0.91\cdot 3n\ln 3n\ln \left( 0.91\cdot 3n\ln 3n\right) }+\frac{1}{2\cdot
0.91\cdot 3n\ln 3n}\right) }.  \label{EQ vars M- bounding martens}
\end{eqnarray}%
The function in Eq. $\left( \ref{EQ vars M- bounding martens}\right) $ is a\
multiple of 2 positive non-increasing functions of $n>1$:\newline
Functions%
\begin{equation*}
\frac{\ln \left( 0.91\cdot 3n\ln 3n\right) }{6\widetilde{b}\ln n}
\end{equation*}%
and 
\begin{equation*}
e^{\gamma \left( \frac{4}{\ln \left( 0.91\cdot 3n\ln 3n+1\right) }+\frac{2}{%
0.91\cdot 3n\ln 3n\ln \left( 0.91\cdot 3n\ln 3n\right) }+\frac{1}{2\cdot
0.91\cdot 3n\ln 3n}\right) }.
\end{equation*}%
Therefore, the function in Eq. $\left( \ref{EQ vars M- bounding martens}%
\right) $ is\ a non-increasing function of $n$ for $n>1$. For $n=64$ the
function in $\left( \ref{EQ vars M- bounding martens}\right) $ is smaller
than $\frac{1}{130}$ and therefore, for any $n\geq 64$ we\ can substitute
the upper bound of $\frac{1}{130}$ into Eq. $\left( \ref{EQ vars M- before
mertens}\right) $:%
\begin{eqnarray*}
&&LB\left( \Delta _{c_{i}^{x}}\left( \underline{p}_{i}+y\right) \right)
-UB\left( \Delta _{c_{i}^{x}}\left( \underline{p}_{i}\right) \right) \\
&>&\frac{a_{n}\alpha _{v}}{\underline{p}_{i}+1}\left(
1-130\dprod\nolimits_{0.91\cdot 3n\ln 3n<p_{\left[ j\right] }\leq n^{6%
\widetilde{b}}}\left( \frac{p_{\left[ j\right] }-1}{p_{\left[ j\right] }}%
\right) \right) \\
&>&\frac{a_{n}\alpha _{v}}{\underline{p}_{i}+1}\left( 1-1\right) =0.
\end{eqnarray*}
\end{proof}

Figure \ref{F3} illustrates the behavior of the lower and upper bounds on $%
\Delta _{c_{i}^{x}}\left( t_{c_{i}^{x}}\right) $ within the range $\left[ 
\underline{p}_{i}-2,\overline{p}_{i}+2\right] $. We arbitrarily chose to
show the bounds for $b_{i}=2.$ The lower bound for $\underline{p}_{i}-1$ and 
$\overline{p}_{i}+1$ is their standalone average cost (depicted by the light
blue line). However, $\underline{p}_{i}+y$ for $0<y<b_{i}$ requires a
tighter bound in order to disprove its optimality (depicted on the pink
line).\FRAME{ftbpFU}{5.8928in}{3.3486in}{0pt}{\Qcb{Lower and upper bounds
for $\Delta _{c_{i}^{x}}\left( t_{c_{i}^{x}}\right) $ in the range $%
t_{c_{i}^{x}}\in \left[ \protect\underline{p}_{i}-1,\overline{p}_{i}+1\right]
,$ depicted for $b_{i}=2.$}}{\Qlb{F3}}{myplot.eps}{\special{language
"Scientific Word";type "GRAPHIC";maintain-aspect-ratio TRUE;display
"USEDEF";valid_file "F";width 5.8928in;height 3.3486in;depth
0pt;original-width 15.9497in;original-height 9.0425in;cropleft "0";croptop
"1";cropright "1";cropbottom "0";filename 'myplot.eps';file-properties
"XNPEU";}}

\begin{theorem}
\label{Theorem vars}In any optimal solution to $\Gamma ,$\emph{\ }$%
t_{c_{i}^{x}}\in \left\{ \underline{p}_{i},\overline{p}_{i}\right\} $ for
any $c_{i}^{x}\in $\emph{Variables.}
\end{theorem}

\begin{proof}
According to Theorem \ref{Theorem vars} the range of the optimal solution
for commodity $c_{i}^{x}$ is $t_{c_{i}^{x}}\in \left[ \underline{p}_{i},%
\overline{p}_{i}\right] $. According to Claim \ref{Claim vars x+ better then
M} the solution $t_{c_{i}^{x}}=\underline{p}_{i}+y$ for $0<y<b_{i}$ costs
more than the solution upper bound on $t_{c_{i}^{x}}\in \left\{ \underline{p}%
_{i},\overline{p}_{i}\right\} .$ Therefore, in any optimal solution to $%
\Gamma ,$\emph{\ }$t_{c_{i}^{x}}\in \left\{ \underline{p}_{i},\overline{p}%
_{i}\right\} .$
\end{proof}

\subsubsection{\label{SEC cycle time clauses}Proof that solving $\Gamma$
optimally is equivalent to solving $\protect\varphi$}

In this section we show that an optimal solution to $\Gamma $ defines an
assignment $\alpha $ to $\varphi $. First we define an assignment $\alpha $
given an optimal solution $S$ to $\Gamma $ as follows: for each commodity $%
c_{i}^{x}\in $\emph{Variables} if the cycle time $t_{c_{i}^{x}}=\underline{p}%
_{i}$ set $\alpha \left( x_{i}\right) =false$. Otherwise, if the cycle time $%
t_{c_{i}^{x}}=\overline{p}_{i}$, set $\alpha \left( x_{i}\right) =true$.
Note that according to Theorem $\ref{Theorem vars}$ $t_{c_{i}^{x}}\in
\left\{ \underline{p}_{i},\overline{p}_{i}\right\} ;$ therefore, these are
the only options. We now want to show that if $\varphi $ is satisfiable then
assignment $\alpha $ that satisfies $\varphi $ gives a solution to $\Gamma $
that is lower than any solution $\alpha ^{\prime }$ that doesn't satisfy%
\emph{\ }$\varphi $. Thus by minimizing $\Gamma $ we solve $\varphi $.

In order to do so we define 3 sets of periods. The first set, denoted by $T^{%
\emph{Constants}}$, includes all the periods in which there is an order of
at least one commodity $c_{lm}^{pv}\in $\emph{Constants}. The second set,
denoted by $T^{\emph{Variables}}$, includes all the periods in which there
is an order of at least one commodity $c_{i}^{x}\in $\emph{Variables}. The
third set, denoted by $T^{\emph{Clauses}}$, includes all the periods in
which there is an order of at least one commodity $c_{r}^{\omega }\in $\emph{%
Clauses}. Accordingly, we formulate the total cost of solution $S$, denoted
by $TC\left( S\right) $, as a sum of 3 cost functions: The first cost
function, $TC_{\emph{Constants}}\left( S\right) $, sums all the costs that
are associated with the commodities $c_{lm}^{pv}\in $\emph{Constants},
including all the joint replenishment costs at periods $t\in T^{\emph{%
Constants}}$. The second cost function, $TC_{\emph{Variables}}\left(
S\right) $, sums all the costs that are associated with the commodities $%
c_{i}^{x}\in $\emph{Variables}, including all the joint replenishment costs
at periods $t\in T^{\emph{Variables}}\backslash T^{\emph{Constants}}$. The
third cost function, $TC_{\emph{Clauses}}\left( S\right) $, sums all the
costs that are associated with the commodities $c_{r}^{\omega }\in $\emph{%
Clauses}, including all the joint replenishment costs at periods $t\in T^{%
\emph{Clauses}}\backslash \left( T^{\emph{Variables}}\cup T^{\emph{Constants}%
}\right) $. Note that 
\begin{eqnarray*}
&&T^{\emph{Constants}}\cup \left( T^{\emph{Variables}}\backslash T^{\emph{%
Constants}}\right) \cup \left( T^{\emph{Clauses}}\backslash \left( T^{\emph{%
Variables}}\cup T^{\emph{Constants}}\right) \right) \\
&=&T^{\emph{Constants}}\cup T^{\emph{Variables}}\cup T^{\emph{Clauses}}
\end{eqnarray*}%
and 
\begin{equation*}
T^{\emph{Constants}}\cap \left( T^{\emph{Variables}}\backslash T^{\emph{%
Constants}}\right) \cap \left( T^{\emph{Clauses}}\backslash \left( T^{\emph{%
Variables}}\cup T^{\emph{Constants}}\right) \right) =\emptyset .
\end{equation*}%
Therefore, 
\begin{equation*}
TC\left( S\right) =TC_{\emph{Constants}}\left( S\right) +TC_{\emph{Variables}%
}\left( S\right) +TC_{\emph{Clauses}}\left( S\right) .
\end{equation*}

According to Theorem \ref{Theorem constants cycle time} the cost $TC_{\emph{%
Constants}}\left( S\right) $ is identical for any optimal solution to $%
\Gamma $.%
\begin{equation*}
TC_{\emph{Constants}}\left( S\right) =\sum_{c_{lm}^{pv}\in \emph{Constants}%
}\left( \frac{K_{c_{lm}^{pv}}}{t_{c_{lm}^{pv}}^{\ast }}+t_{c_{lm}^{pv}}^{%
\ast }\cdot h_{c_{lm}^{pv}}\right) +K_{0}\cdot \left( 1-\alpha _{c}\right) .
\end{equation*}

We now bound the cost function $TC_{\emph{Variables}}\left( S\right) $. We
denote $\Delta _{c_{i}^{x}}^{TC_{\emph{Variables}}}\left(
t_{c_{i}^{x}},S\right) $ as the marginal average periodic cost of the
function $TC_{\emph{Variables}}\left( S\right) $ associated with commodity $%
c $'s cycle time $t_{c_{i}^{x}}$, where $c_{i}^{x}\in $\emph{Variables, }$%
t_{c_{i}^{x}}\in \left\{ \underline{p}_{i},\overline{p}_{i}\right\} $ and a
solution $S$ that applies the characteristics of an optimal solution in
Theorems \ref{Theorem constants cycle time}-\ref{Theorem vars} to the other
commodities in the system. In the next claim we formulate bounds on $\Delta
_{c_{i}^{x}}^{TC_{\emph{Variables}}}\left( t_{c_{i}^{x}},S\right) $. Note
that 
\begin{equation*}
\Delta _{c_{i}^{x}}^{TC_{\emph{Variables}}}\left( t_{c_{i}^{x}},S\right) =%
\frac{K_{c_{i}^{x}}}{t_{c_{i}^{x}}}+t_{c_{i}^{x}}h_{c_{i}^{x}}+K_{0}\cdot
jr^{TC_{\emph{Variables}}}\left( t_{c_{i}^{x}},S\right)
\end{equation*}%
where $jr^{TC_{\emph{Variables}}}\left( t_{c_{i}^{x}},S\right) $ is the
proportion of periods in $T^{\emph{Variables}}\backslash T^{\emph{Constants}%
} $ in which there is an order only of commodity $t_{c_{i}^{x}}$. According
to Lemma \ref{Lemma Polynomial VP} the cycle time $t_{c_{i}^{x}}$ is not a\
multiple of any other cycle time $t_{c_{j}^{x}}$ for any $c_{j}^{x}\in $%
\emph{Variables}. Note that cycle time $t_{c_{j}^{x}}$ is a prime number for
any $c_{j}^{x}\in $\emph{Variables} and therefore $t_{c_{i}^{x}}$ and $%
t_{c_{j}^{x}}\,$ do not share factors. According to Theorem \ref{Theorem
constants cycle time} there is no commodity with a cycle time that is a
factor of $t_{c_{i}^{x}}$ in $\emph{Constants}$. However, for each $%
t_{c_{j}^{x}}\in \left\{ \underline{p}_{j},\overline{p}_{j}\right\} $\ and a
prime number $p\in $\emph{PP}, there is a commodity $c_{lm}^{pv}\in
Constants $\ with a cycle time $c_{lm}^{pv}=p\cdot t_{c_{i}^{x}}$, which
means that at least $\alpha _{c}=\dprod\nolimits_{p\in \text{\emph{PP}}%
}\left( \frac{p-1}{p}\right) $\ of the periods in $TC_{\emph{Variables}}$\
associated with $t_{c_{i}^{x}}$\ are covered by $T^{\emph{Constants}}.$%
{\LARGE \ }Therefore, for a solution $t_{c_{i}^{x}}\in \left\{ \underline{p}%
_{i},\overline{p}_{i}\right\} $, $jr^{TC_{\emph{Variables}}}\left(
t_{c_{i}^{x}},S\right) $ is given by:%
\begin{equation}
jr^{TC_{\emph{Variables}}}\left( t_{c_{i}^{x}},S\right) =\frac{1}{%
t_{c_{i}^{x}}}\cdot \alpha _{c}\dprod\nolimits_{i\neq j,t_{c_{i}^{x}}\in 
\emph{S}}\left( \frac{t_{c_{i}^{x}}-1}{t_{c_{i}^{x}}}\right) .
\label{EQ bound on jr}
\end{equation}%
Therefore 
\begin{equation}
\Delta _{c_{i}^{x}}^{TC_{\emph{Variables}}}\left( t_{c_{i}^{x}},S\right) =%
\frac{K_{c_{i}^{x}}}{t_{c_{i}^{x}}}+t_{c_{i}^{x}}h_{c_{i}^{x}}+K_{0}\cdot 
\frac{1}{t_{c_{i}^{x}}}\cdot \alpha _{c}\dprod\nolimits_{i\neq
j,t_{c_{i}^{x}}\in \emph{S}}\left( \frac{t_{c_{i}^{x}}-1}{t_{c_{i}^{x}}}%
\right) .  \label{delta tx for TC variables}
\end{equation}

\begin{claim}
\label{claim vars p-<p+}For each $c_{i}^{x}$ $\in $\emph{Variables}, and for
each optimal solution $S$, $\Delta_{c_{i}^{x}}^{TC_{\emph{Variables}}}\left( 
\underline{p}_{i},S\right) \leq \Delta_{c_{i}^{x}}^{TC_{\emph{Variables}%
}}\left( \overline{p}_{i},S\right).$
\end{claim}

\begin{proof}
Using Eq. $\left( \ref{delta tx for TC variables}\right) $ we get:%
\begin{eqnarray*}
\Delta _{c_{i}^{x}}^{TC_{\emph{Variables}}}\left( \underline{p}_{i},S\right)
&=&\frac{K_{c_{i}^{x}}}{\underline{p}_{i}}+\underline{p}%
_{i}h_{c_{i}^{x}}+K_{0}\cdot \frac{1}{\underline{p}_{i}}\cdot \alpha
_{c}\dprod\nolimits_{j\neq i,t_{c_{j}^{x}}\in \emph{S}}\left( \frac{%
t_{c_{j}^{x}}-1}{t_{c_{j}^{x}}}\right) ; \\
\Delta _{c_{i}^{x}}^{TC_{\emph{Variables}}}\left( \overline{p}_{i},S\right)
&=&\frac{K_{c_{i}^{x}}}{\overline{p}_{i}}+\overline{p}%
_{i}h_{c_{i}^{x}}+K_{0}\cdot \frac{1}{\overline{p}_{i}}\cdot \alpha
_{c}\dprod\nolimits_{j\neq i,t_{c_{j}^{x}}\in \emph{S}}\left( \frac{%
t_{c_{j}^{x}}-1}{t_{c_{j}^{x}}}\right) .
\end{eqnarray*}%
Substituting for $\overline{p}_{i}=\underline{p}_{i}+b_{i}$ we get: 
\begin{equation*}
\Delta _{c_{i}^{x}}^{TC_{\emph{Variables}}}\left( \overline{p}_{i},S\right) =%
\frac{K_{c_{i}^{x}}}{\left( \underline{p}_{i}+b_{i}\right) }+\left( 
\underline{p}_{i}+b_{i}\right) h_{c_{i}^{x}}+K_{0}\cdot \frac{1}{\left( 
\underline{p}_{i}+b_{i}\right) }\cdot \alpha _{c}\dprod\nolimits_{j\neq
i,t_{c_{j}^{x}}\in \emph{S}}\left( \frac{t_{c_{j}^{x}}-1}{t_{c_{j}^{x}}}%
\right) .
\end{equation*}%
Therefore 
\begin{eqnarray}
&&\Delta _{c_{i}^{x}}^{TC_{\emph{Variables}}}\left( \overline{p}%
_{i},S\right) -\Delta _{c_{i}^{x}}^{TC_{\emph{Variables}}}\left( \underline{p%
}_{i},S\right)  \notag \\
&=&b_{i}\left( h_{c_{i}^{x}}-\frac{K_{c_{i}^{x}}}{\underline{p}_{i}\left( 
\underline{p}_{i}+b_{i}\right) }-\frac{K_{0}}{\underline{p}_{i}\left( 
\underline{p}_{i}+b_{i}\right) }\cdot \alpha _{c}\dprod\nolimits_{j\neq
i,t_{c_{j}^{x}}\in \emph{S}}\left( \frac{t_{c_{j}^{x}}-1}{t_{c_{j}^{x}}}%
\right) \right) .  \label{EQ delta p+ - delta p- before k and h}
\end{eqnarray}%
Substituting for $h_{c_{i}^{x}},$ $K_{c_{i}^{x}}$ and $K_{0}$ using Eqs. $%
\left( \ref{DEF vars h}\right) $, $\left( \ref{DEF vars K}\right) $, and $%
\left( \ref{DEF K_0}\right) $ into Eq. $\left( \ref{EQ delta p+ - delta p-
before k and h}\right) $, we get:%
\begin{eqnarray}
&&\Delta _{c_{i}^{x}}^{TC_{\emph{Variables}}}\left( \overline{p}%
_{i},S\right) -\Delta _{c_{i}^{x}}^{TC_{\emph{Variables}}}\left( \underline{p%
}_{i},S\right)  \notag \\
&=&b_{i}\left( \frac{\alpha _{c}\overline{\alpha }_{v}}{\underline{p}%
_{i}\left( \underline{p}_{i}+b_{i}-1\right) }-\frac{1}{\underline{p}%
_{i}\left( \underline{p}_{i}+b_{i}\right) }\cdot \alpha
_{c}\dprod\nolimits_{j\neq i,t_{c_{j}^{x}}\in \emph{S}}\left( \frac{%
t_{c_{j}^{x}}-1}{t_{c_{j}^{x}}}\right) \right) .
\label{EQ delta p+ - delta p- unbounded}
\end{eqnarray}

Recall that $\forall j:\underline{p}_{j}\leq t_{c_{j}^{x}}\leq \overline{p}%
_{j}$; thus, 
\begin{equation*}
\forall j:\frac{\underline{p}_{j}-1}{\underline{p}_{j}}\leq \frac{%
t_{c_{j}^{x}}-1}{t_{c_{j}^{x}}}\leq \frac{\overline{p}_{j}-1}{\overline{p}%
_{j}}.
\end{equation*}%
Therefore, 
\begin{eqnarray*}
\underline{\alpha}_{v}\cdot \frac{\underline{p}_{i}}{\underline{p}_{i}-1}
&=&\dprod\nolimits_{i\neq j,c_{j}^{x}\in \emph{Variable}}\left( \frac{%
\underline{p}_{j}-1}{\underline{p}_{j}}\right) \\
&\leq &\dprod\nolimits_{j\neq i,t_{c_{j}^{x}}\in \emph{S}}\left( \frac{%
t_{c_{j}^{x}}-1}{t_{c_{j}^{x}}}\right) \\
&\leq &\dprod\nolimits_{i\neq j,c_{j}^{x}\in \emph{Variable}}\left( \frac{%
\overline{p}_{j}-1}{\overline{p}_{j}}\right)=\overline{\alpha }_{v}\cdot 
\frac{\underline{p}_{i}+b_{i}}{\underline{p}_{i}+b_{i}-1},
\end{eqnarray*}%
Substituting for $\dprod\nolimits_{j\neq i,t_{c_{j}^{x}}\in \emph{S}}\left( 
\frac{t_{c_{j}^{x}}-1}{t_{c_{j}^{x}}}\right) $ into Eq. $\left( \ref{EQ
delta p+ - delta p- unbounded}\right)$ we get 
\begin{eqnarray*}
&&\Delta _{c_{i}^{x}}^{TC_{\emph{Variables}}}\left( \overline{p}%
_{i},S\right) -\Delta _{c_{i}^{x}}^{TC_{\emph{Variables}}}\left( \underline{p%
}_{i},S\right) \\
&\geq &b_{i}\left( \frac{\alpha_{c}\overline{\alpha }_{v}}{\underline{p}%
_{i}\left( \underline{p}_{i}+b_{i}-1\right) }-\frac{1}{\underline{p}%
_{i}\left( \underline{p}_{i}+b_{i}\right) }\cdot \alpha _{c}\overline{\alpha 
}_{v}\cdot \frac{\underline{p}_{i}+b_{i}}{\underline{p}_{i}+b_{i}-1}\right)
=0.
\end{eqnarray*}
\end{proof}

We can now lower bound $TC_{\emph{Variables}}\left( S\right) $ by the
solution in which $\forall c_{j}^{x}\in $\emph{Variables}$:t_{c_{j}^{x}}=%
\underline{p}_{j}$. Similarly we can upper bound $TC_{\emph{Variables}%
}\left( S\right) $ by the solution in which $\forall c_{j}^{x}\in $\emph{%
Variables}$:t_{c_{j}^{x}}=\overline{p}_{j}$. The costs of the lower and
upper bounds on $TC_{\emph{Variables}}\left( S\right) $ are given by: \emph{%
\ } 
\begin{eqnarray}
UB\left( TC_{\emph{Variables}}\right) &=&\sum_{c_{i}^{x}\in \emph{Variables}%
}\left( \frac{K_{c_{i}^{x}}}{\overline{p}_{i}}+\overline{p}_{i}\cdot
h_{c_{i}^{x}}\right) +K_{0}\cdot \alpha _{c}\cdot \left( 1-\overline{\alpha }%
_{v}\right) ;  \label{EQ UB(Tvar)} \\
LB\left( TC_{\emph{Variables}}\right) &=&\sum_{c_{i}^{x}\in \emph{Variables}%
}\left( \frac{K_{c_{i}^{x}}}{\underline{p}_{i}}+\underline{p}_{i}\cdot
h_{c_{i}^{x}}\right) +K_{0}\cdot \alpha _{c}\cdot \left( 1-\underline{\alpha 
}_{v}\right) .  \label{EQ LB(Tvar)}
\end{eqnarray}

Last, according to Theorem \ref{Theorem clauses cycle time}, the cycle time
of any commodity $c_{r}^{\omega }\in $\emph{Clauses} where $\omega
_{r}=\left( z_{i}\cup z_{j}\cup z_{s}\right) $ is $t_{c_{r}^{\omega }}^{\ast
}=P\left( z_{i}\right) \cdot P\left( z_{j}\right) \cdot P\left( z_{s}\right) 
$ and $P\left( z_{i}\right) ,P\left( z_{j}\right) ,P\left( z_{s}\right) $
are prime numbers; thus the only factors of $t_{c_{r}^{\omega }}^{\ast }$
are $P\left( z_{i}\right) ,P\left( z_{j}\right) ,P\left( z_{s}\right) $.
Moreover, the cycle times of the commodities $c_{r}^{\omega }\in $\emph{%
Clauses }are not a\ multiple of one another, nor are they a\ multiple of any
cycle time of any commodity $c_{lm}^{pv}\in $\emph{Constants}.

We examine 2 scenarios. In the first there is a commodity with a cycle time
that is a factor of $t_{c_{r}^{\omega }}^{\ast }$. Without loss of
generality assume that there is a commodity with cycle time $P\left(
z_{i}\right) $. Note that according to Theorems \ref{Theorem constants cycle
time}, \ref{Theorem clauses cycle time}, and \ref{Theorem vars} the only
commodity that might have a cycle time of $P\left( z_{i}\right) $ in an
optimal solution is commodity $c_{i}^{x}$. If commodity $c_{i}^{x}$ has a
cycle time of $P\left( z_{i}\right) $, then in the assignment $\alpha $ the
value of the literal $z_{i}\,$\ is true. In this case the clause $\omega
_{r}=\left( z_{i}\cup z_{j}\cup z_{s}\right) $ is satisfied under $\alpha $.
Note that if there is another commodity with cycle time that is a factor of $%
t_{c_{r}^{\omega }}^{\ast }$, then there will be no additional cost for the
joint replenishment. In the second scenario there is no commodity with a
cycle time that is a factor of $t_{c_{r}^{\omega }}^{\ast }$. In this case
we know that the clause $\omega _{r}=\left( z_{i}\cup z_{j}\cup z_{s}\right) 
$ is unsatisfied under $\alpha $. To lower bound the marginal joint
replenishment cost we perform a similar analysis to the one in Eq.\ $\left( %
\ref{EQ bound on jr}\right) $. Yielding the proportion of periods in $T^{%
\emph{Clauses}}\backslash \left( T^{\emph{Variables}}\cup T^{\emph{Constants}%
}\right) $ in which there is an order only of commodity $c_{r}^{\omega }$,
given by $jr^{TC_{\emph{Clauses}}}\left( t_{c_{r}^{\omega }},S\right) $ 
\begin{equation*}
jr^{TC_{\emph{Clauses}}}\left( t_{c_{r}^{\omega }},S\right) =\alpha
_{c}\cdot \dprod\nolimits_{c_{i}^{x}\in \emph{Variable}}\left( \frac{%
t_{c_{i}^{x}}-1}{t_{c_{i}^{x}}}\right)
\end{equation*}%
where $t_{c_{i}^{x}}$ is the cycle time of commodity $c_{i}^{x}$ in the
solution $S$. Since $c_{i}^{x}\in \emph{Variable}:t_{c_{i}^{x}}\in \left\{ 
\underline{p}_{i},\overline{p}_{i}\right\} $, $t_{c_{i}^{x}}$ is minimal at $%
t_{c_{i}^{x}}=\underline{p}_{i}$, and therefore a lower bound on the
marginal joint replenishment cost for any optimal solution is 
\begin{equation*}
LB\left( jr^{TC_{\emph{Clauses}}}\left( t_{c_{r}^{\omega }}\right) \right)
=\alpha _{c}\cdot \underline{\alpha }_{v}.
\end{equation*}%
We denote the group of all the clauses that are not satisfied under $\alpha $
by\ $F$. We can now formulate the cost $TC_{\emph{Clauses}}\left( S\right) $
as: 
\begin{eqnarray}
&&TC_{\emph{Clauses}}\left( S\right) =\sum_{c_{r}^{\omega }\in \emph{Clauses}%
}\left( \frac{K_{c_{r}^{\omega }}}{t_{c_{r}^{\omega }}^{\ast }}%
+t_{c_{r}^{\omega }}^{\ast }\cdot h_{c_{r}^{\omega }}\right) +  \notag \\
&&K_{0}\cdot \alpha _{c}\cdot \dprod\nolimits_{c_{i}^{x}\in \emph{Variable}%
}\left( \frac{t_{c_{i}^{x}}-1}{t_{c_{i}^{x}}}\right) \cdot \left(
1-\dprod\nolimits_{c_{r}^{\omega }\in F}\frac{t_{c_{r}^{\omega }}^{\ast }-1}{%
t_{c_{r}^{\omega }}^{\ast }}\right)  \label{EQ TC(clauses)}
\end{eqnarray}%
(Note that by definition,\ the product of an empty set equals 1). The lower
bound on $TC_{\emph{Clauses}}\left( S\right) $ is given by: 
\begin{eqnarray}
&&LB\left( TC_{\emph{Clauses}}\left( S\right) \right) =\sum_{c_{r}^{\omega
}\in \emph{Clauses}}\left( \frac{K_{c_{r}^{\omega }}}{t_{c_{r}^{\omega
}}^{\ast }}+t_{c_{r}^{\omega }}^{\ast }\cdot h_{c_{r}^{\omega }}\right) + 
\notag \\
&&K_{0}\cdot \alpha _{c}\cdot \underline{\alpha }_{v}\cdot \left(
1-\dprod\nolimits_{c_{r}^{\omega }\in F}\frac{t_{c_{r}^{\omega }}^{\ast }-1}{%
t_{c_{r}^{\omega }}^{\ast }}\right) .  \label{EQ LB(TC(clauses))}
\end{eqnarray}

Eq. $\left( \ref{EQ LB(TC(clauses))}\right) $\ grows with the number of
unsatisfied clauses; thus, in order to show that if $\varphi $ is
satisfiable then any solution that doesn't satisfy $\varphi $ costs more
than a solution that does, it is sufficient to show that the lower bound on
a solution $S$\ in which there is only one unsatisfied clause costs more
than the upper bound on a solution $S^{\prime }$\ that satisfies all the
clauses. Without loss of generality assume that the unsatisfied clause is $%
\omega _{r}$. 
\begin{eqnarray*}
LB\left( TC\left( S\right) \right) &=&TC_{\emph{Constants}}\left( S\right)
+LB\left( TC_{\emph{Variables}}\left( S\right) \right) +LB\left( TC_{\emph{%
Clauses}}\left( S\right) \right) ; \\
UB\left( TC\left( S^{\prime }\right) \right) &=&TC_{\emph{Constants}}\left(
S^{\prime }\right) +UB\left( TC_{\emph{Variables}}\right) +TC_{\emph{Clauses}%
}\left( S^{\prime }\right) .
\end{eqnarray*}%
Note that $TC_{\emph{Constants}}\left( S\right) =TC_{\emph{Constants}}\left(
S^{\prime }\right) $\ is a constant unaffected by the assignment $\alpha $\
and that under the assumption that $S^{\prime }$\ satisfies $\varphi $\ so
does $TC_{\emph{Clauses}}\left( S^{\prime }\right) .$\ For the remaining
cost elements we use upper and lower bounds.{\Huge \ }We now show that $%
UB\left( TC\left( S^{\prime }\right) \right) <LB\left( TC\left( S\right)
\right) .$%
\begin{eqnarray}
&&LB\left( TC\left( S\right) \right) -UB\left( TC\left( S^{\prime }\right)
\right)  \notag \\
&=&TC_{\emph{Constants}}\left( S\right) +LB\left( TC_{\emph{Variables}%
}\left( S\right) \right) +LB\left( TC_{\emph{Clauses}}\left( S\right) \right)
\notag \\
&&-TC_{\emph{Constants}}\left( S^{\prime }\right) -UB\left( TC_{\emph{%
Variables}}\right) -TC_{\emph{Clauses}}\left( S^{\prime }\right)  \notag \\
&=&LB\left( TC_{\emph{Variables}}\right) -UB\left( TC_{\emph{Variables}%
}\right) +LB\left( TC_{\emph{Clauses}}\left( S\right) \right) -TC_{\emph{%
Clauses}}\left( S^{\prime }\right) >0.
\label{EQ vars vs. clauses LB(S)-UB(S')}
\end{eqnarray}

In order to analyze the expression in Eq. $\left( \ref{EQ vars vs. clauses
LB(S)-UB(S')}\right) $, we prove the following Claims:

\begin{claim}
\label{LB(TCvariables)-UB(TCvariables)}$LB\left( TC_{\emph{Variables}%
}\right) -UB\left(TC_{\emph{Variables}}\right) >-b_{i}^{2}\alpha_{c}%
\overline{\alpha}_{v}\underline{p}_{1}^{\frac{1}{3\widetilde{b}}-4}$
\end{claim}

\begin{proof}
According to Eqs. $\left( \ref{EQ UB(Tvar)}\right) $ and $\left( \ref{EQ
LB(Tvar)}\right) :$ 
\begin{eqnarray*}
&&LB\left( TC_{\emph{Variables}}\right) -UB\left( TC_{\emph{Variables}%
}\right) \\
&=&\sum_{c_{i}^{x}\in \emph{Variables}}\left( \frac{K_{c_{i}^{x}}}{%
\underline{p}_{i}}+\underline{p}_{i}\cdot h_{c_{i}^{x}}\right) +K_{0}\cdot
\alpha _{c}\cdot \left( 1-\underline{\alpha }_{v}\right) \\
&&-\sum_{c_{i}^{x}\in \emph{Variables}}\left( \frac{K_{c_{i}^{x}}}{\overline{%
p}_{i}}+\overline{p}_{i}\cdot h_{c_{i}^{x}}\right) -K_{0}\cdot \alpha
_{c}\cdot \left( 1-\overline{\alpha }_{v}\right) .
\end{eqnarray*}%
Substituting for $\overline{p}_{i}=\left( \underline{p}_{i}+b_{i}\right) $ 
\begin{eqnarray}
&&LB\left( TC_{\emph{Variables}}\right) -UB\left( TC_{\emph{Variables}%
}\right)  \notag \\
&=&\sum_{c_{i}^{x}\in \emph{Variables}}\left( \frac{K_{c_{i}^{x}}}{%
\underline{p}_{i}}+\underline{p}_{i}\cdot h_{c_{i}^{x}}-\frac{K_{c_{i}^{x}}}{%
\left( \underline{p}_{i}+b_{i}\right) }-\left( \underline{p}%
_{i}+b_{i}\right) \cdot h_{c_{i}^{x}}\right) +K_{0}\cdot \alpha _{c}\cdot
\left( \overline{\alpha }_{v}-\underline{\alpha }_{v}\right)  \notag \\
&=&\sum_{c_{i}^{x}\in \emph{Variables}}\left( \frac{b_{i}K_{c_{i}^{x}}}{%
\underline{p}_{i}\left( \underline{p}_{i}+b_{i}\right) }-b_{i}\cdot
h_{c_{i}^{x}}\right) +K_{0}\cdot \alpha _{c}\cdot \left( \overline{\alpha }%
_{v}-\underline{\alpha }_{v}\right) .
\label{LB(TC_Variables(S))-UB(TC_Variables)}
\end{eqnarray}%
Substituting for $K_{c_{i}^{x}}$ and $K_{0}$ using Eqs. $\left( \ref{DEF
vars K}\right) $ and $\left( \ref{DEF K_0}\right) $ into Eq. $\left( \ref%
{LB(TC_Variables(S))-UB(TC_Variables)}\right) $ we get:%
\begin{eqnarray}
&&LB\left( TC_{\emph{Variables}}\right) -UB\left( TC_{\emph{Variables}%
}\right)  \notag \\
&=&\sum_{c_{i}^{x}\in \emph{Variables}}\left( -\frac{\alpha _{c}\overline{%
\alpha }_{v}b_{i}}{\underline{p}_{i}\left( \underline{p}_{i}+b_{i}-1\right) }%
\right) +\alpha _{c}\cdot \left( \overline{\alpha }_{v}-\underline{\alpha }%
_{v}\right)  \label{LB(TC_Variables(S))-UB(TC_Variables) before lambda}
\end{eqnarray}%
In order to simplify the expression we define\ for each $c_{i}^{x}\in $\emph{%
Variables}%
\begin{equation*}
\delta _{i}=\frac{\frac{\underline{p}_{i}-1}{\underline{p}_{i}}}{\frac{%
\overline{p}_{i}-1}{\overline{p}_{i}}}.
\end{equation*}%
Substituting for $\overline{p}_{i}=\left( \underline{p}_{i}+b_{i}\right) $%
\begin{eqnarray*}
\delta _{i} &=&\frac{\frac{\underline{p}_{i}-1}{\underline{p}_{i}}}{\frac{%
\underline{p}_{i}+b_{i}-1}{\underline{p}_{i}+b_{i}}}=\frac{\left( \underline{%
p}_{i}-1\right) \left( \underline{p}_{i}+b_{i}\right) }{\underline{p}%
_{i}\left( \underline{p}_{i}+b_{i}-1\right) }=1-\frac{b_{i}}{\underline{p}%
_{i}\left( \underline{p}_{i}+b_{i}-1\right) } \\
&\Rightarrow &\frac{\overline{p}_{i}-1}{\overline{p}_{i}}\delta _{i}=\frac{%
\underline{p}_{i}-1}{\underline{p}_{i}}.
\end{eqnarray*}%
Note that substituting $\frac{\underline{p}_{i}-1}{\underline{p}_{i}}$ into
Eq. $\left( \ref{a_v}\right) $ using Eq. $\left( \ref{a-v}\right) $ we get
that: 
\begin{equation}
\overline{\alpha }_{v}\dprod\nolimits_{c_{i}^{x}\in \emph{Variables}}\delta
_{i}=\underline{\alpha }_{v}.  \label{alphav+ * lambda = alphav-}
\end{equation}%
We substitute this expression into Eq. $\left( \ref%
{LB(TC_Variables(S))-UB(TC_Variables) before lambda}\right) $ 
\begin{eqnarray}
&&LB\left( TC_{\emph{Variables}}\right) -UB\left( TC_{\emph{Variables}%
}\right)  \notag \\
&=&\alpha _{c}\sum_{c_{i}^{x}\in \emph{Variables}}\left( -\frac{b_{i}%
\overline{\alpha }_{v}}{\underline{p}_{i}\left( \underline{p}%
_{i}+b_{i}-1\right) }\right) +\alpha _{c}\cdot \left( \overline{\alpha }_{v}-%
\overline{\alpha }_{v}\dprod\nolimits_{c_{i}^{x}\in \emph{Variables}}\delta
_{i}\right)  \notag \\
&=&\alpha _{c}\left( \sum_{c_{i}^{x}\in \emph{Variables}}\left( -\frac{b_{i}%
\overline{\alpha }_{v}}{\underline{p}_{i}\left( \underline{p}%
_{i}+b_{i}-1\right) }\right) +\overline{\alpha }_{v}\left(
1-\dprod\nolimits_{c_{i}^{x}\in \emph{Variables}}\delta _{i}\right) \right) .
\label{LB(TC_Variables(S))-UB(TC_Variables) after lambda}
\end{eqnarray}%
Note that 
\begin{eqnarray}
&&\dprod\nolimits_{c_{i}^{x}\in \emph{Variables}}\delta
_{i}=\dprod\nolimits_{c_{i}^{x}\in \emph{Variables}}\left( 1-\frac{b_{i}}{%
\underline{p}_{i}\left( \underline{p}_{i}+b_{i}-1\right) }\right)  \notag \\
&=&1-\sum_{c_{i}^{x}\in \emph{Variables}}\frac{b_{i}}{\underline{p}%
_{i}\left( \underline{p}_{i}+b_{i}-1\right) }+  \notag \\
&&\sum_{\substack{ c_{i}^{x}\in \emph{Variables}  \\ c_{j}^{x}\in \emph{%
Variables}  \\ i\neq j}}\left( \frac{b_{i}}{\underline{p}_{i}\left( 
\underline{p}_{i}+b_{i}-1\right) }\cdot \frac{b_{j}}{\underline{p}_{j}\left( 
\underline{p}_{j}+b_{j}-1\right) }\right) -...+\dprod\nolimits_{c_{i}^{x}\in 
\emph{Variables}}\left( \frac{b_{i}}{\underline{p}_{i}\left( \underline{p}%
_{i}+b_{i}-1\right) }\right) .  \label{pi (lambda) equal}
\end{eqnarray}%
We denote the series in Eq. $\left( \ref{pi (lambda) equal}\right) $ as $%
a_{1},a_{2},...,a_{n}$, where $a_{1}=1,a_{2}=-\sum_{c_{i}^{x}\in \emph{%
Variables}}\frac{b_{i}}{\underline{p}_{i}\left( \underline{p}%
_{i}+b_{i}-1\right) },...,a_{n}=\dprod\nolimits_{c_{i}^{x}\in \emph{Variables%
}}\left( \frac{b_{i}}{\underline{p}_{i}\left( \underline{p}%
_{i}+b_{i}-1\right) }\right) $. $a_{3}$ is a sum of $n\left( n-1\right) $
elements where each element is a multiplication $\frac{b_{i}}{\underline{p}%
_{i}\left( \underline{p}_{i}+b_{i}-1\right) }\cdot \frac{b_{j}}{\underline{p}%
_{j}\left( \underline{p}_{j}+b_{j}-1\right) }$ for each combination $i\neq
j,c_{i}^{x},c_{j}^{x}\in \emph{Variables}$. Similarly, $a_{l}$ is a sum of $%
\binom{n}{l-1}$ elements for each combination of $l-1$ commodities in $\emph{%
Variables}$ where each element is a\ multiple of $\left( l-1\right) $
elements of the form $\dprod \left( \frac{b_{i}}{\underline{p}_{i}\left( 
\underline{p}_{i}+b_{i}-1\right) }\right) $. Each element in $a_{l}$ is at
least $\frac{\underline{p}_{1}\left( \underline{p}_{1}+b_{i}-1\right) }{b_{i}%
}$ times bigger than any element in $a_{l+1}$; however, there are $\left(
n-l\right) $ times more elements in $a_{l+1}$ than in $a_{l}$. Since $\frac{%
\underline{p}_{1}\left( \underline{p}_{1}+b_{i}-1\right) }{b_{i}\left(
n-l\right) }>1$, we have $\left\vert a_{l}\right\vert >\left\vert
a_{l+1}\right\vert $. Therefore, we can upper bound the series in Eq. $%
\left( \ref{pi (lambda) equal}\right) $%
\begin{eqnarray}
&&\dprod\nolimits_{c_{i}^{x}\in \emph{Variables}}\delta _{i}  \notag \\
&<&1-\sum_{c_{i}^{x}\in \emph{Variables}}\frac{b_{i}}{\underline{p}%
_{i}\left( \underline{p}_{i}+b_{i}-1\right) }+\sum_{\substack{ c_{i}^{x}\in 
\emph{Variables}  \\ c_{j}^{x}\in \emph{Variables}  \\ i\neq j}}\left( \frac{%
b_{i}}{\underline{p}_{i}\left( \underline{p}_{i}+b_{i}-1\right) }\cdot \frac{%
b_{j}}{\underline{p}_{j}\left( \underline{p}_{j}+b_{j}-1\right) }\right) .
\label{UB for delta_i}
\end{eqnarray}

Since $\forall i:\underline{p}_{1}\leq \underline{p}_{i}$ and $\forall
i:b\geq b_{i}>1$, we can upper bound the second summation in Eq. $\left( \ref%
{UB for delta_i}\right) $ by replacing $\underline{p}_{i}$ and $\underline{p}%
_{j}$ with $\underline{p}_{1}$ and $b_{i}$, $b_{j}$ in the numerator with $b$%
, and in the denominator with $1$. Therefore 
\begin{eqnarray}
&&\dprod\nolimits_{c_{i}^{x}\in \emph{Variables}}\delta _{i}  \notag \\
&<&1-\sum_{c_{i}^{x}\in \emph{Variables}}\frac{b_{i}}{\underline{p}%
_{i}\left( \underline{p}_{i}+b_{i}-1\right) }+\sum_{\substack{ c_{i}^{x}\in 
\emph{Variables}  \\ c_{j}^{x}\in \emph{Variables}  \\ i\neq j}}\left( \frac{%
b}{\underline{p}_{1}\left( \underline{p}_{1}+1-1\right) }\cdot \frac{b}{%
\underline{p}_{1}\left( \underline{p}_{1}+1-1\right) }\right)  \notag \\
&<&1-\sum_{c_{i}^{x}\in \emph{Variables}}\frac{b_{i}}{\underline{p}%
_{i}\left( \underline{p}_{i}+b_{i}-1\right) }+b^{2}n^{2}\cdot \left( \frac{1%
}{\underline{p}_{1}^{2}}\right) ^{2}  \notag \\
&=&1-\sum_{c_{i}^{x}\in \emph{Variables}}\frac{b_{i}}{\underline{p}%
_{i}\left( \underline{p}_{i}+b_{i}-1\right) }+\frac{b^{2}n^{2}}{\underline{p}%
_{1}^{4}}.  \label{pi(lambda)}
\end{eqnarray}%
According to Condition \ref{Cond 2.5}, $n<\underline{p}_{1}^{\frac{1}{6%
\widetilde{b}}}$, replacing $n$ with the upper bound of $\underline{p}_{1}^{%
\frac{1}{6\widetilde{b}}}$ in Eq. $\left( \ref{pi(lambda)}\right) $,%
\begin{eqnarray*}
&&\dprod\nolimits_{c_{i}^{x}\in \emph{Variables}}\delta
_{i}<1-\sum_{c_{i}^{x}\in \emph{Variables}}\frac{b_{i}}{\underline{p}%
_{i}\left( \underline{p}_{i}+b_{i}-1\right) }+\frac{b^{2}\underline{p}_{1}^{%
\frac{2}{6\widetilde{b}}}}{\underline{p}_{1}^{4}} \\
&=&1-\sum_{c_{i}^{x}\in \emph{Variables}}\frac{b_{i}}{\underline{p}%
_{i}\left( \underline{p}_{i}+b_{i}-1\right) }+\frac{b^{2}}{\underline{p}%
_{1}^{4-\frac{1}{3\widetilde{b}}}}.
\end{eqnarray*}%
Replacing $\dprod\nolimits_{c_{i}^{x}\in \emph{Variables}}\delta _{i}$ into
Eq. $\left( \ref{LB(TC_Variables(S))-UB(TC_Variables) after lambda}\right) $%
\begin{eqnarray*}
&&LB\left( TC_{\emph{Variables}}\right) -UB\left( TC_{\emph{Variables}%
}\right) \\
&>&\alpha _{c}\left( \sum_{c_{i}^{x}\in \emph{Variables}}\left( -\frac{b_{i}%
\overline{\alpha }_{v}}{\underline{p}_{i}\left( \underline{p}%
_{i}+b_{i}-1\right) }\right) +\overline{\alpha }_{v}\left( 1-\left(
1-\sum_{c_{i}^{x}\in \emph{Variables}}\frac{b_{i}}{\underline{p}_{i}\left( 
\underline{p}_{i}+b_{i}-1\right) }+\frac{b^{2}}{\underline{p}_{1}^{4-\frac{1%
}{3\widetilde{b}}}}\right) \right) \right) \\
&=&\alpha _{c}\overline{\alpha }_{v}\left( \sum_{c_{i}^{x}\in \emph{Variables%
}}\left( -\frac{b_{i}}{\underline{p}_{i}\left( \underline{p}%
_{i}+b_{i}-1\right) }\right) +\left( \sum_{c_{i}^{x}\in \emph{Variables}}%
\frac{b_{i}}{\underline{p}_{i}\left( \underline{p}_{i}+b_{i}-1\right) }-%
\frac{b^{2}}{\underline{p}_{1}^{4-\frac{1}{3\widetilde{b}}}}\right) \right)
\\
&=&-b^{2}\alpha _{c}\overline{\alpha }_{v}\underline{p}_{1}^{\frac{1}{3%
\widetilde{b}}-4}.
\end{eqnarray*}
\end{proof}

\begin{claim}
$LB\left( TC_{\emph{Clauses}}\left( S\right) \right) -TC_{\emph{Clauses}%
}\left( S^{\prime }\right) >b^{2}\alpha_{c}\overline{\alpha}_{v}\underline{p}%
_{1}^{\frac{1}{3\widetilde{b}}-4}$
\end{claim}

\begin{proof}
Using Eqs. $\left( \ref{EQ TC(clauses)}\right) $ and $\left( \ref{EQ
LB(TC(clauses))}\right) $ with $K_{0}=0$ (see Eq. $\left( \ref{DEF K_0}%
\right) $) we get: 
\begin{eqnarray*}
TC_{\emph{Clauses}}\left( S^{\prime }\right) &=&\sum_{c_{r}^{\omega }\in 
\emph{Clauses}}\left( \frac{K_{c_{r}^{\omega }}}{t_{c_{r}^{\omega }}^{\ast }}%
+t_{c_{r}^{\omega }}^{\ast }\cdot h_{c_{r}^{\omega }}\right) \\
LB\left( TC_{\emph{Clauses}}\left( S\right) \right) &=&\sum_{c_{r}^{\omega
}\in \emph{Clauses}}\left( \frac{K_{c_{r}^{\omega }}}{t_{c_{r}^{\omega
}}^{\ast }}+t_{c_{r}^{\omega }}^{\ast }\cdot h_{c_{r}^{\omega }}\right) + \\
&&\alpha _{c}\cdot \underline{\alpha }_{v}\cdot \left(
1-\dprod\nolimits_{c_{r}^{\omega }\in F}\frac{t_{c_{r}^{\omega }}^{\ast }-1}{%
t_{c_{r}^{\omega }}^{\ast }}\right) .
\end{eqnarray*}%
Since $F=\left\{ c_{r}^{\omega }\right\} $ and as defined in Eq. $\left( \ref%
{DEF clauses t*}\right) ,$ $t_{c_{r}^{\omega }}^{\ast }=P\left( z_{i}\right)
\cdot P\left( z_{j}\right) \cdot P\left( z_{s}\right) $, we get: 
\begin{eqnarray*}
&&LB\left( TC_{\emph{Clauses}}\left( S\right) \right) \\
&=&\sum_{c_{r}^{\omega }\in \emph{Clauses}}\left( \frac{K_{c_{r}^{\omega }}}{%
t_{c_{r}^{\omega }}^{\ast }}+t_{c_{r}^{\omega }}^{\ast }\cdot
h_{c_{r}^{\omega }}\right) +\alpha _{c}\cdot \underline{\alpha }_{v}\frac{1}{%
t_{c_{r}^{\omega }}^{\ast }} \\
&=&\sum_{c_{r}^{\omega }\in \emph{Clauses}}\left( \frac{K_{c_{r}^{\omega }}}{%
t_{c_{r}^{\omega }}^{\ast }}+t_{c_{r}^{\omega }}^{\ast }\cdot
h_{c_{r}^{\omega }}\right) +\alpha _{c}\cdot \underline{\alpha }_{v}\frac{1}{%
P\left( z_{i}\right) \cdot P\left( z_{j}\right) \cdot P\left( z_{s}\right) }.
\end{eqnarray*}%
Therefore,%
\begin{eqnarray*}
&&TC_{\emph{Clauses}}\left( S\right) -TC_{\emph{Clauses}}\left( S^{\prime
}\right) \\
&=&\alpha _{c}\cdot \underline{\alpha }_{v}\frac{1}{P\left( z_{i}\right)
\cdot P\left( z_{j}\right) \cdot P\left( z_{s}\right) }.
\end{eqnarray*}%
Since $\forall i:P\left( z_{i}\right) <\overline{p}_{n}<B\underline{p}_{1}$,
we can lower bound this expression by replacing $P\left( z_{i}\right)
,P\left( z_{j}\right) $ and $P\left( z_{s}\right) $ with $B\underline{p}_{1}$%
. Therefore 
\begin{eqnarray*}
&&TC_{\emph{Clauses}}\left( S\right) -TC_{\emph{Clauses}}\left( S^{\prime
}\right) \\
&>&\alpha _{c}\cdot \underline{\alpha }_{v}\frac{1}{B^{3}\underline{p}%
_{1}^{3}}.
\end{eqnarray*}%
Recall that according to Eq. $\left( \ref{alphav+ * lambda = alphav-}\right) 
$, $\overline{\alpha }_{v}\dprod\nolimits_{c_{i}^{x}\in \emph{Variables}%
}\delta _{i}=\underline{\alpha }_{v}$ and therefore, 
\begin{eqnarray*}
&&TC_{\emph{Clauses}}\left( S\right) -TC_{\emph{Clauses}}\left( S^{\prime
}\right) \\
&>&\alpha _{c}\cdot \overline{\alpha }_{v}\dprod\nolimits_{c_{i}^{x}\in 
\emph{Variables}}\delta _{i}\cdot \frac{1}{B^{3}\underline{p}_{1}^{3}}.
\end{eqnarray*}%
We can lower bound $\dprod\nolimits_{c_{i}^{x}\in \emph{Variables}}\delta
_{i}$ using Eq. $\left( \ref{pi (lambda) equal}\right) $,%
\begin{eqnarray*}
&&TC_{\emph{Clauses}}\left( S\right) -TC_{\emph{Clauses}}\left( S^{\prime
}\right) \\
&>&\alpha _{c}\cdot \overline{\alpha }_{v}\left( 1-\sum_{c_{i}^{x}\in \emph{%
Variables}}\frac{b_{i}}{\underline{p}_{i}\left( \underline{p}%
_{i}+b_{i}-1\right) }\right) \cdot \frac{1}{B^{3}\underline{p}_{1}^{3}}.
\end{eqnarray*}%
Since $\forall i:\underline{p}_{i}\geq \underline{p}_{1}$ and $\forall
i:b_{i}<b$, we can lower bound this expression by replacing $\underline{p}%
_{i}$ with $\underline{p}_{1}$ and $b_{i}$ in the numerator with $b$ and in
the denominator with $1$. Therefore:%
\begin{eqnarray*}
&&TC_{\emph{Clauses}}\left( S\right) -TC_{\emph{Clauses}}\left( S^{\prime
}\right) \\
&>&\alpha _{c}\cdot \overline{\alpha }_{v}\left( 1-\frac{bn}{\underline{p}%
_{1}^{2}}\right) \cdot \frac{1}{B^{3}\underline{p}_{1}^{3}}.
\end{eqnarray*}%
Substituting for $B$ according to Condition \ref{Cond 2} we have:%
\begin{eqnarray}
&&TC_{\emph{Clauses}}\left( S\right) -TC_{\emph{Clauses}}\left( S^{\prime
}\right)  \notag \\
&>&\frac{\alpha _{c}\cdot \overline{\alpha }_{v}}{\left( 6\widetilde{b}\log
n\right) ^{3\widetilde{b}}\underline{p}_{1}^{3}}\left( 1-\frac{bn}{%
\underline{p}_{1}^{2}}\right) >\frac{\alpha _{c}\cdot \overline{\alpha }_{v}%
}{n^{3\widetilde{b}}\underline{p}_{1}^{3}}\left( 1-\frac{bn}{\underline{p}%
_{1}^{2}}\right) ,  \label{TC_Clauses(S)-TC_Clauses(S') with n}
\end{eqnarray}%
were the second inequality holds for any $n>3536$ even for the upper bound
on $\widetilde{b}$ attained by the \emph{Polymath8} project \cite{PM2015}.
According to Condition \ref{Cond 2.5}, $n<\underline{p}_{1}^{\frac{1}{6%
\widetilde{b}}}$, replacing $n$ with the upper bound $\underline{p}_{1}^{%
\frac{1}{6\widetilde{b}}}$ into Eq. $\left( \ref%
{TC_Clauses(S)-TC_Clauses(S') with n}\right) $ we get,%
\begin{eqnarray*}
&&TC_{\emph{Clauses}}\left( S\right) -TC_{\emph{Clauses}}\left( S^{\prime
}\right) \\
&>&\frac{\alpha _{c}\cdot \overline{\alpha }_{v}}{\underline{p}_{1}^{3.5}}%
\left( 1-\frac{b}{\underline{p}_{1}^{2-\frac{1}{6\widetilde{b}}}}\right) >%
\frac{\alpha _{c}\cdot \overline{\alpha }_{v}}{\underline{p}_{1}^{3.5}}%
\left( 1-\frac{b}{\underline{p}_{1}^{2-\frac{1}{6\widetilde{b}}}}\right) \\
&>&\frac{\alpha _{c}\cdot \overline{\alpha }_{v}}{\underline{p}_{1}^{3.5}}%
\left( \frac{b^{2}}{\underline{p}_{1}^{0.5-\frac{1}{6\widetilde{b}}}}\right)
=\frac{b^{2}\alpha _{c}\overline{\alpha }_{v}}{\underline{p}_{1}^{4-\frac{1}{%
3\widetilde{b}}}},
\end{eqnarray*}%
where the last inequality holds for $\underline{p}_{1}>2362$ ($n\geq 2$
according to Condition \ref{Cond 2.5}) even for the upper bound on $%
\widetilde{b}$ attained by the \emph{Polymath8} project \cite{PM2015}.
\end{proof}

Therefore, 
\begin{eqnarray*}
&&LB\left( TC\left( S\right) \right) -UB\left( TC\left( S^{\prime }\right)
\right) \\
&=&TC_{\emph{Constants}}\left( S\right) +LB\left( TC_{\emph{Variables}%
}\left( S\right) \right) \\
&&+TC_{\emph{Clauses}}\left( S\right) -TC_{\emph{Constants}}\left( S^{\prime
}\right) -UB\left( TC_{\emph{Variables}}\right) -TC_{\emph{Clauses}}\left(
S^{\prime }\right) \\
&=&LB\left(TC_{\emph{Variables}}\left( S\right) \right) -UB\left(TC_{\emph{%
Variables}}\right) +TC_{\emph{Clauses}}\left( S\right) -TC_{\emph{Clauses}%
}\left( S^{\prime }\right) \\
&>&\frac{b^{2}\alpha _{c}\overline{\alpha }_{v}}{\underline{p}_{1}^{4-\frac{2%
}{3\widetilde{b}}}}-\frac{b^{2}\alpha _{c}\overline{\alpha }_{v}}{\underline{%
p}_{1}^{4-\frac{2}{3\widetilde{b}}}}=0.
\end{eqnarray*}

\begin{conclusion}
A solution to $\Gamma $ that reflects an assignment $\alpha $ that satisfies 
$\varphi $ costs less than any solution to $\Gamma $ that reflects an
assignment $\alpha ^{\prime }$ that does not satisfy $\varphi $. Therefore,
solving \emph{PJRP }is at least as hard as solving \emph{$3SAT.$}
\end{conclusion}

\section{\label{section - finite}$\mathcal{NP}$-Hardness of the Periodic
Joint Replenishment Problem with finite horizon}

The model of the finite time horizon is similar to the model of the infinite
time horizon; however, since the time horizon is finite it is possible that
for a commodity $c$ the last cycle will not be a whole one. In the finite
model we assume a time horizon of $T$ periods. Similarly to the infinite
time horizon, we analyzed the standalone \emph{problem }cost function. In a
case that $\func{mod}\left( T,t_{c}\right) =0\,$, for a commodity $c,$ the
last cycle is a full one. In this case, the average periodic cost, denoted
by $\tilde{g}\left( t_{c}\right) $ is equal to $g\left( t_{c}\right) $. The
expression for the average periodic cost as a function of the cycle time $%
t_{c}$ in the case $\func{mod}\left( T,t_{c}\right) \neq 0$ is a complex
one. In order to avoid this situation for the same reduction defined in
Section \ref{section- reduction}, we define 
\begin{equation*}
T=\dprod\nolimits_{i=2}^{\left( \overline{p}_{n}\right) ^{3}}i,
\end{equation*}

This guarantees that all the cycle times that were analyzed in Section \ref%
{section - optimality} are of the form $t_{c}$ where $\func{mod}\left(
T,t_{c}\right) =0$. Therefore, the observations from Section \ref{section -
optimality} apply for the finite horizon model. However, since $T$ is not
polynomial in $n$, the problem is $\mathcal{NP}$-hard but not necessarily
strongly $\mathcal{NP}$-hard.

\section{\label{section - Summary}Summary}

In this paper we answer the long-standing open question regarding the
computational complexity of \emph{PJRP} with integer cycle times for a
finite time horizon as well as for an infinite time horizon. We provided a
proof that \emph{PJRP} with integer cycle times and an infinite time horizon
is strongly $\mathcal{NP}$-hard and that \emph{PJRP} with integer cycle
times and a finite time horizon is $\mathcal{NP}$-hard.

Another important problems yet to be answered is defining the computational
complexity of \emph{PJRP} with non-integer cycle times and of the strict 
\emph{PJRP}.

\begin{acknowledgement}
We thank Professor Retsef Levi for challenging us to solve the long-lasting
question of the complexity of the PJRP.
\end{acknowledgement}

\begin{acknowledgement}
We thank and acknowledge the help of Noga Glassner in preparing this paper.
\end{acknowledgement}

\section*{Appendix: Proof of Lemma \protect\ref{Lemma Polynomial VP}.}

\begin{proof}
In our proof we show that\ there is a set of at least $nB$ pairs of
consecutive primes that satisfy Conditions \ref{Cond 1} and \ref{Cond 2}
within an interval $\left[ \underline{p}_{1},B\underline{p}_{1}\right] $ for
some $B$ and $\underline{p}_{1}$ that satisfy Condition \ref{Cond 2.5}.
Then, we show that there is a subset of at least $n$ pairs that satisfy
Condition \ref{Cond 3} as well. Finally, we show that this set could be
identified in $O\left( n^{6\widetilde{b}+1}\log ^{\widetilde{b}}n\right) $
time.

According to \cite{Z2013} there are at least $\frac{B\underline{p}_{1}}{\log
^{\widetilde{b}}B\underline{p}_{1}}$ $2-$tuple$\left( b\right) $ prime
pairs; hence, there are at least $\frac{B\underline{p}_{1}}{\log ^{%
\widetilde{b}}B\underline{p}_{1}}$ pairs of consecutive primes with a gap of
at most $b$ between them. In order to lower bound the number of $2-$tuple$%
\left( b\right) $ primes in an interval $\left[ \underline{p}_{1},B%
\underline{p}_{1}\right] $ we need an upper bound for the number of $2-$tuple%
$\left( b\right) $ primes smaller than $\underline{p}_{1}.$ We use an
extremely un-tight bound, assuming all primes smaller than $\underline{p}%
_{1} $ are $2-$tuple$\left( b\right) $ primes. According to \cite{RS1962}
there are at most $\frac{1.255\underline{p}_{1}}{\log \underline{p}_{1}}$
prime numbers smaller than $\underline{p}_{1}$. Hence, there are no more
than $0.5\frac{1.255\underline{p}_{1}}{\log \underline{p}_{1}}$ prime pairs
smaller than $\underline{p}_{1}$. Therefore, the number of $2-$tuple$\left(
b\right) $ prime pairs in an interval $\left[ \underline{p}_{1},B\cdot 
\underline{p}_{1}\right] $, denoted $N_{VP}$, satisfies: 
\begin{equation*}
N_{VP}\geq \frac{B\underline{p}_{1}}{\log ^{\widetilde{b}}\left( B\underline{%
p}_{1}\right) }-\frac{0.63\underline{p}_{1}}{\log \underline{p}_{1}}.
\end{equation*}%
Note that 
\begin{equation*}
B\underline{p}_{1}=\underline{p}_{1}{}^{\log _{\underline{p}_{1}}B}\cdot 
\underline{p}_{1}=\underline{p}_{1}^{1+\log _{\underline{p}_{1}}B};
\end{equation*}%
therefore, 
\begin{eqnarray}
N_{VP} &\geq &\frac{B\underline{p}_{1}}{\log ^{\widetilde{b}}\left( 
\underline{p}_{1}^{1+\log _{\underline{p}_{1}}B}\right) }-\frac{0.63%
\underline{p}_{1}}{\log \underline{p}_{1}}  \notag \\
&=&\frac{B\underline{p}_{1}}{\left( 1+\log _{\underline{p}_{1}}B\right) ^{%
\widetilde{b}}\log ^{\widetilde{b}}\underline{p}_{1}}-\frac{0.63\underline{p}%
_{1}}{\log \underline{p}_{1}}  \notag \\
&=&\frac{\underline{p}_{1}}{\log \underline{p}_{1}}\left( \frac{B}{\left(
1+\log _{\underline{p}_{1}}B\right) ^{\widetilde{b}}\log ^{\widetilde{b}-1}%
\underline{p}_{1}}-0.63\right) .  \label{Eq pi_z(x)}
\end{eqnarray}%
We set 
\begin{eqnarray}
B &=&\log ^{\widetilde{b}}\underline{p}_{1}  \label{B} \\
&\Longrightarrow &B^{\frac{\widetilde{b}-1}{\widetilde{b}}}=\log ^{%
\widetilde{b}-1}\underline{p}_{1}  \notag
\end{eqnarray}%
Note that: 
\begin{equation*}
\left( 1+\log _{\underline{p}_{1}}B\right) =\left( 1+\widetilde{b}\log _{%
\underline{p}_{1}}\log \underline{p}_{1}\right) <1+\widetilde{b}.
\end{equation*}%
Thus, for a large enough $x_{s}$ we have the numerator in Eq. $\left( \ref%
{Eq pi_z(x)}\right) $ satisfy:%
\begin{equation*}
\left( 1+\log _{\underline{p}_{1}}B\right) ^{\widetilde{b}}\log ^{\widetilde{%
b}-1}\underline{p}_{1}<\left( 1+\widetilde{b}\right) B^{\frac{\widetilde{b}-1%
}{\widetilde{b}}}<B
\end{equation*}%
and%
\begin{equation}
\frac{B}{\left( 1+\log _{\underline{p}_{1}}B\right) ^{\widetilde{b}}\log ^{%
\widetilde{b}-1}\underline{p}_{1}}>1  \label{eq B}
\end{equation}%
Substituting Eq. $\left( \ref{eq B}\right) $ into Eq. $\left( \ref{Eq
pi_z(x)}\right) $ we get:%
\begin{equation*}
N_{VP}>0.37\frac{\underline{p}_{1}}{\log \underline{p}_{1}}
\end{equation*}%
Next, we need to find $\underline{p}_{1}$ and show that it is not greater
than $n^{6\widetilde{b}}.$ That is, we need to find $\underline{p}_{1}$ such
that $\left[ \underline{p}_{1},B\underline{p}_{1}\right] $ contains at least 
$2Bn$ pairs of $2-$tuple$\left( b\right) $ primes:%
\begin{equation}
N_{VP}>0.37\frac{\underline{p}_{1}}{\log \underline{p}_{1}}>2Bn
\label{Eq 2Bn}
\end{equation}%
We substitute for $B$ using Eq. $\left( \ref{B}\right) $ and, in order to
satisfy Condition \ref{Cond 2.5}, we replace $n$ with the upper bound of $%
\underline{p}_{1}^{\frac{1}{6\widetilde{b}}};$ hence the condition in Eq. $%
\left( \ref{Eq 2Bn}\right) $ maintains that:%
\begin{eqnarray*}
0.37\frac{\underline{p}_{1}}{\log \underline{p}_{1}} &>&2\underline{p}_{1}^{%
\frac{1}{6\widetilde{b}}}\log ^{\widetilde{b}}\underline{p}_{1} \\
\underline{p}_{1}^{1-\frac{1}{6\widetilde{b}}} &>&5.4\log ^{\widetilde{b}+1}%
\underline{p}_{1}
\end{eqnarray*}%
Let us extract the base $2$ logarithm of both sides of the inequality:%
\begin{eqnarray*}
\left( 1-\frac{1}{6\widetilde{b}}\right) \log \underline{p}_{1} &>&\left( 
\widetilde{b}+1\right) \left( \log \log \underline{p}_{1}\right) +2.435 \\
\frac{\log \underline{p}_{1}}{\log \log \underline{p}_{1}} &>&\frac{%
\widetilde{b}+1+\frac{2.435}{\left( \log \log \underline{p}_{1}\right) }}{1-%
\frac{1}{6\widetilde{b}}}
\end{eqnarray*}%
Condition \ref{Cond 2.5} guaranties that for $n>2$ we have $\log \log 
\underline{p}_{1}>2.435$; hence, $\frac{2.435}{\left( \log \log \underline{p}%
_{1}\right) }<1.$ Let us look at the right side of the inequality:%
\begin{equation*}
\frac{\widetilde{b}+1+\frac{2.435}{\left( \log \log \underline{p}_{1}\right) 
}}{1-\frac{1}{6\widetilde{b}}}<\frac{\widetilde{b}+2}{1-\frac{1}{6\widetilde{%
b}}}=\frac{6\widetilde{b}^{2}+12\widetilde{b}}{6\widetilde{b}-1}<\frac{%
\left( 6\widetilde{b}-1\right) \left( \widetilde{b}+3\right) }{6\widetilde{b}%
-1}=\widetilde{b}+3
\end{equation*}%
Thus, it is sufficient to find a $\underline{p}_{1}$ that satisfies:%
\begin{eqnarray*}
&&\frac{\log \underline{p}_{1}}{\log \log \underline{p}_{1}}>\widetilde{b}+3.
\\
&\Rightarrow &\underline{p}_{1}>\log ^{\widetilde{b}+3}\underline{p}_{1}
\end{eqnarray*}%
in order to satisfy Conditions \ref{Cond 1}-\ref{Cond 2.5}. hence, for a
sufficiently large $n$ (and according to Conditions \ref{Cond 2.5} a
sufficiently large $\underline{p}_{1}$) the condition in Eq. $\left( \ref{Eq
2Bn}\right) $ is satisfied. That is, setting 
\begin{eqnarray*}
&&\underline{p}_{1}>n^{6\widetilde{b}} \\
&\Rightarrow &B=\log ^{\widetilde{b}}\underline{p}_{1}>\left( 6\widetilde{b}%
\log n\right) ^{\widetilde{b}}
\end{eqnarray*}%
for a sufficiently large $n$ guaranties that there are at least $Bn$ pairs
of $2-$tuple$\left( b\right) $ primes within $\left[ \underline{p}_{1},B%
\underline{p}_{1}\right] ;$ hence, there are at least $Bn$ pairs of
consecutive primes with a gap of at most $b$ between them.\footnote{%
The size of $n$ required to satisfy this constraint for the values of $b$
and $\widetilde{b}$ found by \cite{PM2015} is extremely large (greater than $%
2^{470}$). However, for any smaller $n$ empirical evidence show that the
twin prime conjecture \cite{Gu1994} and the first Hardy-Littlewood
conjecture \cite{HL1923} hold (Hardy and Wright \cite{HR1979} note that "the
evidence, when examined in detail, appears to justify the conjecture," and
Shanks \cite{S1993} stated "the evidence is overwhelming"). Using twin
primes the $n$ required drops dramatically. Thus, our reduction is
permissible for any $n\geq 12$ (the whole analysis and bounds become
tighter).}

If we chose the set $\emph{VP}_{2}$ from within $\left[ \underline{p}_{1},B%
\underline{p}_{1}\right] $ all pairs would satisfy Conditions \ref{Cond 2}
and \ref{Cond 2.5}.

Let us now greedily chose pairs of consecutive primes that satisfy Condition %
\ref{Cond 1} starting from the smallest pair greater than $\underline{p}_{1}$
and adding them to $\emph{VP}_{2}$ and $\emph{VP}$ as long as Condition \ref%
{Cond 3} is satisfied with respect to the elements already chosen to be in $%
\emph{VP}$.

Since $\underline{p}_{1}\leq B\underline{p}_{1}$, each prime in $\emph{VP}$
can be a factor of at most $B-1$ numbers within $\left[ \underline{p}_{1},B%
\underline{p}_{1}\right] .$ Since $\left\vert \emph{VP}\right\vert =2n$
there are at most $2n\left( B-1\right) $ numbers that are factored by an
element in $\emph{VP;}$ thus there are at most $2n\left( B-1\right) $ pairs
of consecutive primes that do not satisfy Condition \ref{Cond 3} and at
least $n$ that do. Hence, the sets $\emph{VP}_{2}$ and $\emph{VP}$ that
satisfy Conditions \ref{Cond 1}-\ref{Cond 3} are found in a range that is
polynomial to $n$.

Next we show that the sets $\emph{VP}_{2}$ and $\emph{VP}$ could be found in
polynomial time by showing that all primes within the range $\left[ 0,B%
\underline{p}_{1}\right] $ could be found in polynomial time.

Using Sieve of Eratosthenes \cite{R1996} the first $k$ prime numbers could
be found in $O\left( k^{2}\right) $ time. Since there is an upper bound of
no more than $1.25506\frac{B\underline{p}_{1}}{\log B\underline{p}_{1}}$
prime in the range $\left[ 0,B\underline{p}_{1}\right] $ (see \cite{RS1962}%
), finding the prime numbers in the range $\left[ 0,B\underline{p}_{1}\right]
$ can be done in 
\begin{equation*}
O\left( \left( \frac{B\underline{p}_{1}}{\log B\underline{p}_{1}}\right)
^{2}\right) ,
\end{equation*}%
which is polynomial. Once the primes are identified, the greedy method to
construct $\emph{VP}_{2}$ and $\emph{VP}$ requires $O\left( nB\underline{p}%
_{1}\right) =O\left( n^{6\widetilde{b}+1}\log ^{\widetilde{b}}n\right) $
time.
\end{proof}

\end{document}